\documentclass[12pt,reqno,centertags]{amsart}
\usepackage{a4}
\usepackage{amsmath,amsthm,amscd,amssymb}
\usepackage{latexsym}
\usepackage{upref}
\usepackage{hyperref}
\usepackage{upgreek}
\usepackage{textgreek}
\usepackage{cite}
\usepackage{tikz}
\usepackage{enumerate}
\usepackage{multicol}
\usepackage{accents}
\usepackage{lmodern}
\usepackage{microtype}
\usepackage{graphicx}
\usepackage{subcaption} 
\usepackage{dsfont}
\usepackage{colonequals}
\usepackage{bbm}
\usepackage{fullpage}
\UseRawInputEncoding

\newcommand{\abs}[1]{\left|#1\right|}
\newcommand\rE[1]{\upharpoonleft_{#1}}
\newcommand\lrb[1]{\left\lbrace#1 \right\rbrace}
\newcommand\lrp[1]{\left(#1 \right)}
\newcommand\lrc[1]{\left[#1 \right]}
\newcommand\hh{{\mathcal H^{\oplus_2}}}
\newcommand\C{{\mathbb C}}
\newcommand\N{{\mathbb N}}
\newcommand\R{{\mathbb R}}
\newcommand\cc[1]{\overline {#1}}
\newcommand\ip[2]{\left\langle {#1},{#2} \right\rangle}
\newcommand\no[1]{\left\| {#1} \right\|}
\newcommand\vE[2]{{\scriptsize\begin{pmatrix}{#1}\\[-.7mm]{#2}\end{pmatrix}}}
\newcommand\oP[1]{{#1}_{\tiny{\odot}}}
\newcommand\dS{{\rm\bf N}}
\DeclareMathOperator{\im}{Im}
\DeclareMathOperator{\re}{Re}
\DeclareMathOperator{\dom}{dom}
\DeclareMathOperator{\ran}{ran}
\DeclareMathOperator{\mul}{mul}
\DeclareMathOperator{\Span}{span}

\allowdisplaybreaks
\numberwithin{equation}{section}

\newtheorem{theorem}{Theorem}[section]
\newtheorem{lemma}[theorem]{Lemma}
\newtheorem{corollary}[theorem]{Corollary}

\newtheorem{proposition}[theorem]{Proposition}
\theoremstyle{definition}
\newtheorem{definition}[theorem]{Definition}

\theoremstyle{remark}
\newtheorem{remark}[theorem]{Remark}

\begin{document}
\title[Semigroups of linear relations]{$C_0$-Semigroups on d.m.o. linear relations}
\author[D. Regalado-Hern\'andez and J. I. Rios-Cangas]{ Luis D. Regalado-Hern\'andez and Josu\'e I. Rios-Cangas}
\address{Departamento de Matem\'aticas, Universidad Aut\'onoma Metropolitana, Iztapalapa Campus,  San Rafael Atlixco 186, 09340 Iztapalapa, Mexico City.}
\email{l.omega0613@gmail.com,\quad jottsmok@xanum.uam.mx}

\date{\today}
\subjclass{Primary 47A06, 47D06;
Secondary 47D60}

\keywords{Linear relations, $C_0$-semigroups, Infinitesimal generators, Anti-accumulative linear relations, Hille-Yosida-Lumer-Phillips Theorem}

\begin{abstract}
This article develops a theory of one-parameter semigroups for linear relations, introducing the notion of domain-multivalued orthogonal (d.m.o.) relations. We establish complete characterizations for generators of uniformly and strongly continuous semigroups, showing them to be maximal and densely d.m.o. relations, respectively. Furthermore, the Hille-Yosida and Lumer-Phillips theorems are extended to this general setting, providing necessary and sufficient conditions for generating $C_0$-semigroups of contractions (identified as maximal anti-accumulative relations) and unitary operators (anti-selfadjoint relations). This rigorous framework sheds new light on the peculiarities of one-parameter semigroups for linear relations.
 \end{abstract}

\maketitle

\section{Introduction}  
The landscape of functional analysis, particularly in the study of Hilbert spaces, frequently encounters mathematical objects that defy the traditional definitions of single-valued, densely defined linear operators. Such scenarios, which arise naturally in the study of singular differential equations, boundary value problems, and quantum mechanics, require the broader framework of \textit{linear relations} (also know as multi-valued operators). Historically, the notion of a linear relation arose from the compelling need to describe the adjoint of a non-densely defined linear operator. John von Neumann pioneered this theory, and his seminal work \cite{MR0034514} is widely recognized as the precursor to the modern treatment of linear relations. This framework was further expanded in \cite{MR0123188,MR1631548}, leading to the sophisticated theories detailed in recent comprehensive accounts such as \cite{MR3971207}.

Within this generalized setting, the concepts of dissipativity and accumulativity are not merely abstract properties; they are crucial for characterizing the physical behavior of complex systems. Their profound relevance is evident across various applications \cite{MR4275499,MR1318517,MR2486805}, particularly in the extension theory of dissipative and accumulative relations as well as operators \cite{MR4082306,MR3971207,MR4302456}. The dissipative, as well as accumulative, properties of a relation are fundamental to understanding the energetics and time-evolution of systems, serving as the very cornerstones for the generation of semigroups. The authoritative monograph by Behrndt, Hassi, and de Snoo \cite{MR3971207} rigorously develops these concepts, leveraging sophisticated tools like boundary triplets and Weyl functions to tackle boundary value problems and extension theory for differential operators.

Classical von Neumann extension theory provides a powerful, elegant method for classifying the self-adjoint extensions of symmetric operators using deficiency indices, a theory that naturally extends to symmetric linear relations \cite{MR0477855, MR0361889}. However, because many real-world physical systems are inherently non-conservative, there is a critical need for a broader extension theory capable of encompassing dissipative and accumulative relations. This expanded theory, which also benefits significantly from the boundary triplet formalism, allows for a comprehensive parameterization of maximally dissipative or anti-dissipative extensions of a given relation. Significant contributions to these generalized extension theories have been made in recent years, notably by R\'ios-Cangas and Silva \cite{MR4082306}, as well as by Behrndt, Hassi, and de Snoo \cite{MR3971207}.

Semigroup theory for linear operators was extensively developed in \cite{MR1503076,MR1503079,MR1503425}, profoundly impacting the study of the time evolution of diverse physical systems, particularly in quantum mechanics.  Building upon the foundational work on operator semigroups \cite{MR710486}, the concept of one-parameter semigroups has been successfully extended to linear relations. Early contributions by Favini and Fuhrman \cite{MR1666189} established crucial convergence results for degenerate evolution equations, while Baskakov \cite{MR2475046} provided a broad framework for analyzing their spectral properties. More recently, Arendt, Chalendar, and Moletsane \cite{MR4815643} generalized fundamental milestones --such as the Hille-Yosida, Lumer-Phillips, and Trotter-Kato theorems-- to $m$-dissipative linear relations. This shift toward multi-valued generators does more than just complement classical operator theory; it captures the inherent peculiarities of these relations, providing a rigorous mathematical framework essential for analyzing open quantum systems, singular differential operators, and complex control systems. 

To systematically establish the theory of semigroups for linear relations, we organize the present article as follows. In Section~\ref{s2}, we briefly review standard facts about linear relations and introduce the pivotal concept of \textit{domain-multivalued orthogonal} (d.m.o.) relations, exploring their deep connection with their operator parts. For instance, Theorem~\ref{th:resolvet-dom-mul} reveals how the powers of a relation's resolvent --and their respective norms-- are inextricably linked to those of its operator part. Section~\ref{s3} is devoted to the classes of dissipative and accumulative relations. These are inherently d.m.o. and can be completely characterized by their operator parts (Theorem~\ref{th:maximal-Ll}). Within this section, we also address symmetric (and particularly, self-adjoint) relations, establishing a similar characterization in Theorem~\ref{rmk:Asym-AA}. In Section~\ref{s4}, we explore anti-dissipative and anti-accumulative relations; since these arise as $-i$ variations of dissipative and accumulative relations (Lemma~\ref{lem:nonA-Diss}), they too are d.m.o. This section establishes in Theorem~\ref{thm:vN-formulae-AS} the von Neumann formulae for anti-symmetric relations, and culminates with a characterization of this class (Theorem~\ref{th:oP-asymm}). Finally, the core objective of this article is realized in Section~\ref{s5}. Here, Definition~\ref{dfn:semigroup} formalizes the concept of a one-parameter semigroup of a linear relation, and Definition~\ref{dfn:generator-of-semigroup} identifies its infinitesimal generator. We prove in Theorem~\ref{th:uniformly-CSG} that a closed relation generates a uniformly continuous semigroup if and only if it is maximal d.m.o. Furthermore, Theorem~\ref{th:resl-C0} dictates that a closed relation generates a strongly continuous semigroup if and only if it is densely d.m.o. and possesses a positive ray in its regular set. The culmination of this framework --the Hille-Yosida and Lumer-Phillips theorems for linear relations-- is presented in Subsection~\ref{ss5.3}, where Theorem~\ref{thm:HY-LP-lr} asserts that the generator of a $C_0$-semigroup with a contractive operator part is precisely a maximal anti-accumulative relation. Moreover, this generator becomes anti-selfadjoint if the operator part of the semigroup is unitary (Corollary~\ref{cor:HY-LP-aA}). By way of illustrating these rich theoretical results, the final section develops examples of $C_0$-semigroup generators constructed as one-dimensional perturbations.

\section{Domain-multivalued orthogonal linear relations}\label{s2}
For the Hilbert space $\hh\colonequals\mathcal H\oplus\mathcal H$, a \emph{linear relation} (relation for short) is a subspace $T\leq\hh$ such that 
\begin{align*}
\dom T&\colonequals\lrb{f\in \mathcal  H\,:\, \vE fg\in T}\,,&
\ran T&\colonequals\lrb{g\in \mathcal  H\,:\, \vE fg\in T}\,,\\
\ker T&\colonequals\lrb{f\in \mathcal  H\,:\, \vE f0\in T}\,,&
\mul T&\colonequals\lrb{g\in \mathcal  H\,:\, \vE 0g\in T}\,,
\end{align*} 
are subspaces in $\mathcal H$ and denote the \emph{domain}, \emph{range}, \emph{kernel} and \emph{multivalued} of $T$, respectively. So, a relation extend the notion of a linear operator $T$ in $\mathcal H$ since the graph of $T$ is a particular case of a linear relation. Indeed, a relation $T$ is an operator (identifying with its graph) if and only if $\mul T=\{0\}$.

The \emph{closure} of $T$ is denoted by $\cc T$. Besides, for $T,S\leq \hh$ and $\alpha\in\C$, we consider the following operations
\begin{align*}
T+S&\colonequals\lrb{\vE f{g+h}\,:\,\vE fg\in T,\,\vE fh\in S}\,,&
\zeta T&\colonequals \lrb{ \vE f{\zeta g}\,:\,\vE fg\in T}\,,\\
ST&\colonequals\lrb{ \vE fk\,:\,\vE fg\in T,\,\vE gk\in S}\,,& T^{-1}&\colonequals\lrb{\vE gf\,:\, \vE fg\in T }\,,
\end{align*}
which turn out relations in $\hh$. 
\begin{remark}\label{rm:LR-distribution} For $n\in\N$ and linear relations $T,T_1,\dots,T_n$  it follows from \cite[Prop.\,I.4.2]{MR1631548} that
\begin{align*}
TT_1+\dots+TT_n\subset T(T_1+\dots+T_n)\qquad\mbox{and}\qquad (T_1+\dots+T_n)T\subset T_1T+\dots+T_nT\,.
\end{align*}
Besides:
\begin{enumerate}[(i)]
\item\label{eq:it1LRD} If $\ran T_j\subset\dom T$, for $j=1,\dots,n$, then $TT_1+\dots+TT_n= T(T_1+\dots+T_n)$.
\item\label{eq:it2LRD} If $T$ is an operator then $(T_1+\dots+T_n)T= T_1T+\dots+T_nT$.
\end{enumerate}
In particular, one has for $\{\alpha_j\}_{j=1}^n\subset \C$ that $\alpha_1T+\dots+\alpha_nT=\lrp{\alpha_1+\dots+\alpha_n}T$.
\end{remark}

For $T, S$ l.i. (linearly independent), i.e, $T\cap S=\lrb{\vE 00}$, it is readily to verify that 
\begin{gather}\label{eq:distributed-inverse}
(T\dotplus S)^{-1}=T^{-1}\dotplus S^{-1}\,.
\end{gather}  
Consequently, $(T\oplus S)^{-1}=T^{-1}\oplus S^{-1}$, if $T\perp S$.

It is well-known that the \emph{adjoint} of $T$ is the closed relation $T^*$ given by
\begin{gather*}
 T^*\colonequals\lrb{\vE hk\in \hh\,:\,\ip kf=\ip hg,\, \forall\, \vE fg\in T}\,,
\end{gather*}
which satisfies
\begin{equation}
\begin{aligned}\label{eq:p-Adjoint}
T^*&=(-T^{-1})^{\perp}\,,&\qquad\qquad S\subset  T&\Rightarrow T^*\subset S^*\,,\\
T^{**}&=\cc T\,,&\qquad\qquad (\alpha T)^*&=\cc{\alpha} T^*\,,\,\mbox{ with } \alpha\neq0\,,\\
(T^*)^{-1}&=(T^{-1})^*\,,&\qquad\qquad \mathcal  H&=\cc{\ran T}\oplus\ker T^*\,.
\end{aligned}
\end{equation}

If $\dom T\subset \dom S$ and $\dom (T+S)^*\subset \dom S^*$, then \cite[Thm.\,3.41]{MR0123188}
\begin{gather*}
(T+S)^*=T^*+S^*\,.
\end{gather*}
 
In this work, a relation $T$ is said to be \emph{bounded} if there exists $C>0$, such that 
\begin{gather*}
\no g\leq C\no f\,,\quad\mbox{for all}\quad\vE fg\in T\,.
\end{gather*}
Thereby, the multivalued of any bounded relation $T$ is trivial, viz. $T$ is a linear operator. Besides, $T+S$ is closed if $T, S$ are closed and $S$ is bounded \cite[Thm.\,3.2.3]{MR1192782}. 

It is convenient to address the \emph{quasi-regular} set of a relation $T$ given by 
\begin{gather*}
\hat\rho(T)\colonequals\lrb{\zeta \in \C\
:\ (T-\zeta I)^{-1}\mbox{ is bounded}}\,,
\end{gather*} 
which is an open set in $\C$. Regarding the class of all bounded operators $\mathcal {B(H)}$ with domain the whole space $\mathcal  H$, we define the \emph{regular} set of $T$ by
\begin{gather*}
\rho(T)\colonequals\lrb{\zeta \in \C\ :\ (T-\zeta I)^{-1}\in\mathcal {B(H)}}\,,
\end{gather*}
which also is open and belongs to $\hat\rho(T)$. For simplicity, the regular set is taken over closed relations, since otherwise it is empty. Also, the following spectral sets of a relation $T$ are regarding in this work:
\begin{align*}
\sigma(T)&\colonequals\C\backslash \rho(T), &\mbox{\it (spectrum)}\\
\hat\sigma(T)&\colonequals\C\backslash \hat\rho(T),&\mbox{\it(spectral core)}\\
\sigma_p(T)&\colonequals\{\zeta \in \C\, :\, \ker (T-\zeta I)\neq \{0\}\},&\mbox{\it(point)}\\
\sigma_d(T)&\colonequals\{\zeta \in \sigma_{p}(T)\, :\, \dim\ker (T-\zeta I)<\infty\},&\mbox{\it(discrete)}\\
\sigma_p^{\infty}(T)&\colonequals\{\zeta \in \sigma_{p}(T)\, :\, \dim\ker (T-\zeta I)=\infty\},&\mbox{\it(point non-discrete)}\\
\sigma_c(T)&\colonequals\{\zeta \in \C\, :\, \ran (T-\zeta I)\neq \overline{\ran (T-\zeta I)}\},&\mbox{\it(continuous)}
\end{align*} 
Thence, it is not hard to show that $\hat\sigma(T)=\sigma_p(T)\cup\sigma_c(T)$.

It is a well-known, and a very useful one, that a closed relation $T$ is decomposed by 
\begin{gather}\label{eq:multivalued-operator.parts}
T=\oP T\oplus T_\infty\,,\quad\mbox{where}\quad T_{\infty}\colonequals\lrb{\vE 0g\in T}\,;\quad \oP T\colonequals T\ominus T_{\infty}\,,
\end{gather}
denote the \emph{multivalued part} and \emph{operator part} of $T$, respectively. The operator part $\oP T$ is a closed linear operator, while the multivalued part $T_\infty$ is a purely multivalued closed relation. The decomposition \eqref{eq:multivalued-operator.parts} allows studying spectral of $T$ by means of $\oP T$.

\begin{lemma}\label{rm:distributionTS}
Let $T,S$ be relations with $\dom{T} \cap \mul{S} = \{0\}$.
\begin{enumerate}[(i)]
	\item\label{distrTS1} If $S$ is closed, then $TS=T \oP S$.
	\item\label{distrTS2} If $T$ is closed, then $TS = \oP T S \oplus T_\infty$. 
\end{enumerate}
\end{lemma}
\begin{proof}
 Since $\dom{T} \cap \mul{S} = \{0\}$, for every $\vE{h}{k} \in TS$, there exist $\vE{h}{f} \in S$ and $\vE{f}{k} \in T$ with $f \in \dom{T} \cap \ran{S}$, which means $\vE{h}{f} \in \oP{S}$, i.e., $\vE{h}{k} \in T\oP{S}$, whence it follows \eqref{distrTS1}. Now, $\dom{T} \cap \mul{S} = \{0\}$ yields $\mul T=\mul TS$. Thus, $\vE{h}{k} \in TS$ is equivalent to the existence of  $\vE{h}{f} \in S$, $\vE{f}{k_1} \in \oP T$ and $\vE{0}{k_2} \in T_\infty$, such that $k = k_1 + k_2$, viz. $\vE{h}{k} = \vE{h}{k_1 + k_2} \in \oP T S \oplus T_\infty$. This prove \eqref{distrTS2}.	
\end{proof}

\begin{definition}\label{df:dmo-T}
We say that a relation $T$ is \emph{domain-multivalued orthogonal} (briefly d.m.o.) if 
\begin{gather*}
\dom T\subset\lrp{\mul T}^\perp\,.
\end{gather*}
In this fashion, we may say that $T$ is \emph{densely d.m.o.}, if $\cc{\dom T}=\lrp{\mul T}^\perp$. Moreover, $T$ is \emph{maximal d.m.o.}, if $\dom T=\lrp{\mul T}^\perp$.
\end{definition}

E.g., every linear operator is d.m.o, and every operator in $\mathcal B(\mathcal H)$ is maximal d.m.o.
 
\begin{theorem}\label{prop:PC-tn}
For $n\in\N$ and closed relations $T_0,T_1,\dots, T_n$ such that $\mul T_j=\mul T_0$, $j=1,\dots,n$,  the following hold:
\begin{enumerate}[(i)]
\item\label{ita:PC} $T_0+\dots +T_n=(\oP {T_0}+\dots +\oP {T_n})\oplus {T_0}_\infty$, wherefrom  $T_0+\dots +T_n$, $\oP {T_0}+\dots +\oP {T_n}$ are simultaneously closed. In which case,
\begin{gather*}
\oP{(T_0+\dots +T_n)}=\oP {T_0}+\dots +\oP {T_n}\qquad\mbox{while}\qquad
(T_0+\dots +T_n)_\infty= {T_0}_\infty\,.
\end{gather*}

\item\label{itb:PC} If $T_1,\dots, T_n$ are d.m.o., then 
\begin{gather}\label{eq:TS-oP-dmo}
T_nT_{n-1}\cdots T_0=T_n\lrp{\oP{T_{n-1}}\cdots \oP{T_0}}=\oP {T_n}\oP{T_{n-1}}\cdots \oP{T_0}\oplus {T_0}_\infty\,,
\end{gather}
whence $T_nT_{n-1}\cdots T_0$ and $\oP{T_{n-1}}\cdots \oP{T_0}$ are simultaneously closed. In this case, 
\begin{gather*}
\oP{(T_nT_{n-1}\cdots T_0)}=\oP {T_n}\oP{T_{n-1}}\cdots \oP{T_0}\qquad\mbox{while}\qquad
(T_nT_{n-1}\cdots T_0)_\infty= {T_0}_\infty\,.
\end{gather*}
\end{enumerate}
\end{theorem}
\begin{proof}
\eqref{ita:PC}: bearing in mind ${T_j}_\infty={T_0}_\infty$ and  \eqref{eq:multivalued-operator.parts}, one has that $\ran \oP {T_j},\ran \oP {T_0}\perp \mul T_0$, wherefrom $(\oP {T_0}+\dots+\oP {T_n}) \perp {T_0}_\infty$. Besides, if $\vE f{h_0+\dots+h_n}\in T_0+\dots +T_n$ then there exist $\vE f{t_s}\in\oP {T_s}$ and $a_s\in\mul T_s$, such that $h_s=t_s+a_s$, with $s=0,\dots,n$. Thus, 
\begin{gather*}
\vE f{h_0+\dots+h_n}=\vE f{t_0+\dots+t_n}+\vE 0{a_0+\dots+a_n}\in (\oP {T_0}+\dots+\oP {T_n})\oplus {T_0}_\infty\,,
\end{gather*}
whence $T_0+\dots +T_n\subset(\oP {T_0}+\dots+\oP {T_n})\oplus {T_0}_\infty$, and thence equal, since 
\begin{gather*}
(\oP {T_0}+\dots+\oP {T_n})\oplus {T_0}_\infty=T_0+\oP {T_1}+\dots+\oP {T_n}\subset T_0+\dots +T_n\,.
\end{gather*}
To conclude, if $T_0+\dots +T_n$ is closed then so is $\oP{(T_0+\dots +T_n)}=\oP {T_0}+\dots+\oP {T_n}$. On the other hand, if $\oP {T_0}+\dots+\oP {T_n}$ is closed then so is $(\oP {T_0}+\dots+\oP {T_n})\oplus {T_0}_\infty=T_0+\dots +T_n$, since ${T_0}_\infty$ is closed.

\eqref{itb:PC}: the d.m.o property and Lemma~\ref{rm:distributionTS}.\eqref{distrTS1} readily implies that $T_nT_{n-1}=T_n\oP{T_{n-1}}$. Besides, 
\begin{gather*}
	\dom{T_{n} \oP{T_{n-1}}\cdots \oP{T_{n-j}}} \subset \dom{T_{n-j}} \perp \mul{T_{n-j}} = \mul{T_{n-j-1}}, \qquad  j=1,\ldots,n-1 \,,
\end{gather*}
where again Lemma~\ref{rm:distributionTS}.\eqref{distrTS1} yields 
\begin{gather*}
T_nT_{n-1}T_{n-2}\cdots T_0=T_n\oP{T_{n-1}}T_{n-2}\cdots T_0=T_n\oP{T_{n-1}}\oP{T_{n-2}}\cdots T_0=\dots=T_n\lrp{\oP{T_{n-1}}\cdots \oP{T_0}}\,,
\end{gather*}
whence it follows the left-hand side of \eqref{eq:TS-oP-dmo}. Now, the right-hand side of \eqref{eq:TS-oP-dmo} is straightforward by above and Lemma~\ref{rm:distributionTS}.\eqref{distrTS2}. The closedness condition is proved along the same lines as in item \eqref{ita:PC}.
\end{proof}

Given $T,S\leq\hh$, the relation $T_S$ is considered in the Hilbert space ${(\mul S)^\perp}^{\oplus_2}$ by
\begin{gather*}
T_S\colonequals T\cap {(\mul S)^\perp}^{\oplus_2}\,.
\end{gather*}
Thus, $T$ and $T_S$ are closed simultaneously. Besides, one simply computes that $(T^{-1})_T=(T_T)^{-1}$. Moreover, when $T$ is closed, $T_T=(\oP T)_T$ and therefore, $T_T$ is a closed operator. Oftentimes $T_S$ is regarding as a relation in $\hh$. 

\begin{remark}\label{rmk:dom-multperp-spectral} A closed d.m.o. $T$ satisfies $\dom \oP T, \ran \oP T\subset (\mul T)^\perp$ and
\begin{gather}\label{eq:Tt-ToP}
T_T=(\oP T)_T=\oP T\cap{(\mul T)^\perp}^{\oplus_2}=\oP T\,.
\end{gather}
Additionally (cf. \cite[Thm.\,2.10]{MR4082306}),
\begin{align}\label{eq:espectral-properties-inherited}
\begin{aligned}
\sigma(T)&=\sigma(T_T)\,,&\hat\sigma(T)&=\hat\sigma(T_T)\,,&\sigma_c(T)&=\sigma_c(T_T)\,,\\
\sigma_p(T)&=\sigma_p(T_T)\,,&\sigma_p^\infty(T)&=\sigma_p^\infty(T_T)\,,&\sigma_d(T)&=\sigma_d(T_T)\,.
\end{aligned}
\end{align}
\end{remark}

Consider the \emph{resolvent} of a relation $T$ given by 
\begin{gather*}
\mathcal R_T(\zeta)=(T-\zeta I)^{-1}\,,\qquad\zeta\in\C\,,
\end{gather*}
which is a relation and bounded if $\zeta\in\hat\rho(T)$. Moreover, $\mathcal R_T(\zeta)$ and $T$ are simultaneously closed and in such a case $\mathcal R_T(\zeta)\subset \mathcal B(\mathcal H)$, for $\zeta\in\rho(T)$.

\begin{remark}
It is not hard to show the following \emph{Hilbert's resolvent identity} for relations,
\begin{gather}\label{eq:Hilbert's-resolvent}
\mathcal R_T(\zeta)-\mathcal R_T(\lambda)=(\zeta-\lambda)\mathcal R_T(\zeta)\mathcal R_T(\lambda)\,,
\end{gather}
wherefrom it follows that $\mathcal R_T(\zeta)$ and $\mathcal R_T(\lambda)$ commute.  
\end{remark}

Similarly to operators (cf. \cite[Sect.\,3.7]{MR1192782}), $\mathcal R_T(\lambda)$ depends analytically on $\lambda \in\rho(T)$. Besides, for $\zeta \in\rho(T)$ such that $\abs{\lambda-\zeta}\leq C_\zeta^{-1}$, where $\no{\mathcal R_T(\zeta)}\leq C_\zeta$, the resolvent $\mathcal R_T(\lambda)$ can be represented by a uniformly convergent power series
\begin{gather*}
\mathcal R_T(\lambda)=\sum_{k=0}^\infty(\lambda-\zeta)^k\mathcal R_T(\zeta)^{k+1}\,.
\end{gather*}
 Moreover, the scalar-valued function $\lambda\mapsto \ip{f}{\mathcal R_T(\lambda) g}$ is analytic on $\rho(T)$, for all $f,g\in\mathcal H$.

It is worth to mention that the Hilbert's resolvent identity \eqref{eq:Hilbert's-resolvent} can be extended to resolvent operators with respect
to a bounded operator \cite[Sect.\,1.11]{MR3971207}.

\begin{theorem}\label{th:resolvet-dom-mul}
For $\zeta\in\C$, the resolvent of a closed d.m.o. $T$ satisfies 
\begin{gather}\label{eq:resolvent-OP}
\lrc{\mathcal R_T(\zeta)}^n=\lrc{\mathcal R_{T_T}(\zeta)}^n\oplus T_{\infty}^{-1}\quad\mbox{and}\quad \lrb{[\mathcal R_T(\zeta)]^n}_T=\lrc{\mathcal R_{T_T}(\zeta)}^n\,,\qquad n\in\N\,.
\end{gather}
Besides, if $ \zeta\in\hat\rho(T)$ then  
\begin{gather}\label{eq:rest-norm}
\no {\lrc{\mathcal R_T(\zeta)}^n}=\no{\lrc{\mathcal R_{T_T}(\zeta)}^n}\,.
\end{gather}
\end{theorem}
\begin{proof}
Since  $\dom T\subset (\mul T)^\perp$, it follows that $T-\zeta I=(\oP T-\zeta I)\oplus T_\infty$. Thus, \eqref{eq:distributed-inverse} and \eqref{eq:Tt-ToP} imply \eqref{eq:resolvent-OP}, for $n=1$ and, hence, for all $n\in\N$. Besides, \eqref{eq:rest-norm} readily follows by the left-hand side of \eqref{eq:resolvent-OP}, since $\mul T\subset \ker \lrc{\mathcal R_T(\zeta)}^n$.
\end{proof}

In the following, we shall discuss several classes of domain-multivalued orthogonal relations.

\section{Dissipative and accumulative linear relations}\label{s3}
Our discussion first addresses the definition of a dissipative linear relation.
\begin{definition}
We call a relation $L$ \emph{dissipative} (\emph{accumulative}) whenever 
\begin{gather*}
\im\ip fg\geq 0\quad (\im\ip fg\leq 0)\,,\qquad\mbox{for all } \vE fg\in L\,.
\end{gather*}
Besides, $L$ is called \emph{maximal dissipative} (\emph{maximal accumulative}) if it is dissipative (accumulative) and does not admit proper dissipative (accumulative) extensions. Consequently, any maximal dissipative (accumulative) relation is closed.
\end{definition}

A key observation is that a linear relation $L$ is dissipative (or maximal dissipative) if and only if $-L$ is accumulative (or maximal accumulative). This equivalence allows us to restrict our attention to formulating results for dissipative relations; the corresponding statements for accumulative relations are direct consequences.

Let us denote the lower and upper open half-planes by 
\begin{gather*}
\C_{-}\colonequals \{\zeta\in\C\,:\,\im \zeta <0\}\quad\mbox{and}\quad\C_{+}\colonequals \{\zeta\in\C\,:\,\im \zeta >0\}\,,\quad\mbox{respectively}\,.
\end{gather*}

The following is adapted from \cite[Sect.\,3]{MR4082306} and, in particular, \cite[Thm.\,3.1]{MR4082306}.

\begin{theorem}\label{th:Dissipative-resolvet} For a linear relation $L$ the following hold:
\begin{enumerate}[(i)]
\item  $L$ is dissipative (maximal dissipative) if and only if $\C_-\subset \hat\rho(L)$ ($L$ is closed and $\C_-\subset\rho(L)$) and $\no {\mathcal R_L(\zeta)}\leq-(\im \zeta)^{-1}$, for $\zeta\in\C_-$.
\item  $L$ is accumulative (maximal accumulative) if and only if  $\C_+\subset \hat\rho(L)$ ($L$ is closed and $\C_+\subset\rho(L)$) and $\no {\mathcal R_L(\zeta)}\leq(\im \zeta)^{-1}$, for $\zeta\in\C_+$.
\end{enumerate}
\end{theorem}

A closed dissipative (accumulative) relation whose domain is the whole space is a bounded maximal dissipative (accumulative) operator \cite[Prop.\,3.4]{MR4082306}. Moreover, the spectrum of a maximal dissipative (maximal accumulative) relation $L$ satisfies (cf. \cite[Prop.\,3.5]{MR4082306})
\begin{gather*}
\rho(L)\cap(\C_-\cup\R)=\hat\rho(L)\cap(\C_-\cup\R)\qquad (\rho(L)\cap(\C_+\cup\R)=\hat\rho(L)\cap(\C_+\cup\R))\,,
\end{gather*}
which is equivalent to 
\begin{gather*}
 \sigma(L)\cup\R=\hat\sigma(L)\cup\R\,.
\end{gather*}

\begin{remark}\label{rm:dom-mul_Diss}
Every dissipative (accumulative) relation is d.m.o. and every maximal dissipative (maximal accumulative) relation is densely d.m.o. \cite[Lem.\,2.1]{MR3057107}. 
In this fashion, every closed dissipative (closed accumulative) relation satisfies the spectral and resolvent properties  \eqref{eq:espectral-properties-inherited}, \eqref{eq:resolvent-OP} and \eqref{eq:rest-norm}.
\end{remark}

We now focus on the operator part of a closed dissipative (closed accumulative) relation.
\begin{lemma}\label{lm:operator-part-L}
Let $L$ be a closed relation. Then, $L$ is dissipative (accumulative) if and only if it is d.m.o. such that $\oP L$ is dissipative (accumulative).
\end{lemma}
\begin{proof}
If $L$ is a closed dissipative relation, then it is d.m.o. by Remark~\ref{rm:dom-mul_Diss} and $\oP L\subset L$ implies that $\oP L$ is dissipative. Conversely, bearing in mind \eqref{eq:multivalued-operator.parts}, for any $\vE {f}{g+h}\in L$, with $\vE fg\in\oP L$ and $h\in\mul L$, the d.m.o. property of $L$ implies that $f\perp h$ and
\begin{gather*}
\im\ip{f}{g+h}=\im\ip fg\geq 0\,,
\end{gather*}
whence it follows that $L$ is dissipative.
\end{proof}
We cannot relax the d.m.o. property of $L$ in the converse statement of Lemma~\ref{lm:operator-part-L}. Of course, the zero operator $0$ on $\mathcal H$ is selfadjoint and thence dissipative. So, for a unit element $u\in\mathcal H$, one has that $\vE u{-iu}\in A=0\oplus\Span\lrb{\vE 0 u}$ and $\im\ip{u}{-iu}=-1$, whence $\oP A=0$ is dissipative but $A$ is not.  

\begin{theorem}\label{th:maximal-Ll}
A closed relation $L$ is maximal dissipative (maximal accumulative) if and only if it is densely d.m.o. with $L_L$ a maximal dissipative (maximal accumulative) operator in ${(\mul L)^\perp}^{\oplus_2}$.
\end{theorem}
\begin{proof}
If $L$ is maximal dissipative then by virtue of Remark~\ref{rm:dom-mul_Diss} it is densely d.m.o., which by Remark~\ref{rmk:dom-multperp-spectral}.\eqref{eq:Tt-ToP} and Lemma~\ref{lm:operator-part-L} the operator $L_L$ is closed dissipative. Besides, Theorems~\ref{th:Dissipative-resolvet} and \ref{th:resolvet-dom-mul}.\eqref{eq:rest-norm}, with \eqref{eq:espectral-properties-inherited}, imply $\C_-\subset\rho(L)=\rho(L_L)$ and $\no {\mathcal R_{L_L}(\zeta)}=\no {\mathcal R_L(\zeta)}\leq-(\im \zeta)^{-1}$, for $\zeta\in\C_-$, viz. $L_L$ is maximal dissipative in ${(\mul L)^\perp}^{\oplus_2}$.

Conversely, Remark~\ref{rmk:dom-multperp-spectral}.\eqref{eq:Tt-ToP} and
 Lemma~\ref{lm:operator-part-L} imply that $L$ is closed dissipative, while by Theorem~	\ref{th:Dissipative-resolvet}, one gets $\C_-\subset\rho(L_L)=\rho(L)$ and $\no {\mathcal R_{L}(\zeta)}=\no {\mathcal R_{L_L}(\zeta)}\leq-(\im \zeta)^{-1}$, for $\zeta\in\C_-$, whence $L$ is maximal dissipative.
\end{proof}

\subsection{Symmetric and selfadjoint linear relations}
The following defines a particular class of dissipative linear relations.
\begin{definition}
A relation $A$ is \emph{symmetric} if $A\subset A^*$ and \emph{selfadjoint} when $A=A^*$. A symmetric relation $A$ is call \emph{maximal}, if it does not admit proper symmetric extensions. So, a maximal symmetric relation is closed
\end{definition}

It is useful to recall the following characterization of symmetric relations, adapted from \cite[Rmk.\,3.2]{MR4082306} and \cite[Prop.\,4.1]{MR4091412}. 
\begin{proposition}\label{pro:characterization-symmetric}
For a linear relation $A$ the following are equivalent: 
\begin{enumerate}[(i)]
\item $A$ is symmetric.
\item $\ip fg\in\R$, for all $\vE fg\in A$.
\item $\ip fk=\ip gh$, for all $\vE fg,\vE hk\in A$.
\item\label{it4:sym} $\C\backslash \R\subset\hat\rho(A)$ and  
\begin{gather*}
\no{\mathcal R_A(\zeta)}\leq \abs{\im\zeta}^{-1}\,,\quad \mbox{for $\zeta\in\C\backslash \R$}\,.
\end{gather*}

\item\label{it5:sym} $A$ is simultaneously dissipative and accumulative.
\end{enumerate}
Moreover, for a closed symmetric relation $A$ the following properties are identical:
\begin{enumerate}[(a)]
\item $A$ is selfadjoint.
\item $\ker(A^*\pm iI)=\{0\}$.
\item $\sigma(A)=\hat\sigma(A)$.
\end{enumerate}
\end{proposition}

Following \cite[Sect.\,2]{MR4082306}, the \emph{deficiency space} and the \emph{deficiency index} of a relation $T$ are  
\begin{gather*}
\dS_\zeta(T)\colonequals\lrb{\vE f{\zeta f}\in T}\quad\mbox{and}\quad \eta_\zeta(T)\colonequals \dim \dS_{\cc\zeta}(T^*)\,,\qquad \mbox{respectively}\,. \quad (\zeta\in\C)
\end{gather*}
When $T$ is closed, $\eta_\zeta(T)$ remains constant  on each connected component of $\hat\rho(T)$. Thereby, from Theorem~\eqref{th:Dissipative-resolvet} one can denote the deficiency index of a closed dissipative relation $L$ and the deficiency index of a closed accumulative relation $\hat L$ by 
\begin{gather*}
\eta_-(L)\colonequals \eta_\zeta(L)\quad\mbox{and}\quad \eta_+(\hat L)\colonequals \eta_{\beta}(\hat L)\,,
\qquad \zeta,-\beta\in\C_-\,.
\end{gather*}
By Proposition~\ref{pro:characterization-symmetric}.\eqref{it4:sym}-\eqref{it5:sym}, a closed symmetric relation $A$ has deficiency indices  
\begin{gather*}
(\eta_+(A),\eta_-(A))\colonequals(\eta_{-\zeta}(A),\eta_{\zeta}(A))\,,\quad \zeta\in\C_-\,,
\end{gather*}
which satisfy $\eta_\pm(-A)=\eta_\mp(A)$.

\begin{remark}\label{rm:dissipative-extension}
The second von Neumann formula \cite[Thm.\,4.7]{MR4082306} guarantees that a closed symmetric relation $A$ has closed dissipative (accumulative) extensions if and only if $\eta_-(A)\neq0$ ($\eta_+(A)\neq0$), and has closed symmetric extensions if and only if $\eta_\pm(A)\neq0$. Moreover, if $L$, $\hat L$ and $S$ are closed dissipative,  accumulative and symmetric extensions of $A$, respectively, then (cf. \cite[Cor.\,4.8]{MR4082306})
\begin{gather*}
\eta_-(A)=\eta_-(L)+\dim[L/A]\,,\quad \eta_+(A)=\eta_+(\hat L)+\dim[\hat L/A]\quad\mbox{and}\\ \eta_\pm(A)=\eta_\pm(S)+\dim[S/A]\,.
\end{gather*}
\end{remark}

\begin{lemma}\label{lem:max-Symmetric}
For a closed symmetric relation $A$ the following are equivalent:
\begin{enumerate}[(i)]
\item\label{it1:max-sym} $A$ is maximal.
\item\label{it2:max-sym} One of the indices $\eta_+(A)$, $\eta_-(A)$ is equal zero.
\item\label{it3:max-sym} $A$ maximal dissipative or maximal accumulative.
\item\label{it4:max-sym} $\C_-$ or $\C_+$ is contained in $\rho(A)$.
\end{enumerate}
\end{lemma}
\begin{proof}
\eqref{it1:max-sym}$\Rightarrow$\eqref{it2:max-sym}: if $A$ is maximal symmetric relation then it has no proper symmetric extensions, which by Remark~\ref{rm:dissipative-extension} item \eqref{it2:max-sym} holds true. 

\eqref{it2:max-sym}$\Rightarrow$\eqref{it3:max-sym}: if $\eta_-(A)=0$ then in view of Remark~\ref{rm:dissipative-extension}, $A$ is a closed dissipative relation with no proper dissipative extensions, i.e., $A$ is maximal dissipative. The case $\eta_+(A)=0$ follows by above, since $\eta_+(A)=\eta_-(-A)$. and $-A$ is symmetric. 

\eqref{it3:max-sym}$\Rightarrow$\eqref{it4:max-sym}: it is straightforward from Theorem~\ref{th:Dissipative-resolvet}, bearing in mind that $\rho(-A)=-\rho(A)$ and $\C_+=-\C_-$.

\eqref{it4:max-sym}$\Rightarrow$\eqref{it1:max-sym}: if $\C_-\subset \rho(A)$ then one obtains by Theorem~\ref{th:Dissipative-resolvet} and Proposition~\ref{pro:characterization-symmetric}.\eqref{it4:sym} that $A$ is maximal dissipative, which by Remark~\ref{rm:dissipative-extension}, one has $\eta_-(A)=0$ and $A$ is maximal symmetric. The case $\C_+\subset \rho(A)$ implies $\C_-\subset\rho(-A)$, which by the above $\eta_+(A)=\eta_-(-A)=0$, viz. $A$ is maxima symmetric. 
\end{proof}

Since a symmetric relation is simultaneously dissipative and accumulative, a maximal symmetric relation may admit either dissipative or accumulative extensions. As an example, Remark~\ref{rm:dissipative-extension} and Lemma~\ref{lem:max-Symmetric} imply that if $A$ and $B$ are closed symmetric relations such that $\eta_+(A)=0, \eta_-(A)\neq 0$ and $\eta_+(B)\neq0, \eta_-(B)= 0$ respectively, then both $A$ and $B$ are maximal symmetric relations. However, $A$ possesses dissipative extensions, while $B$ admits accumulative ones.

\begin{remark}\label{rmk:Asym-AA}
By virtue of Remark~\ref{rm:dom-mul_Diss} and Proposition~\ref{pro:characterization-symmetric}.\eqref{it5:sym}, every symmetric relation is d.m.o., while every maximal symmetric relation is densely d.m.o. Thereby, every closed symmetric relation fulfills  the spectral properties \eqref{eq:espectral-properties-inherited} as well as the resolvent properties  \eqref{eq:resolvent-OP} and \eqref{eq:rest-norm}.
 \end{remark}
 
 \begin{theorem}\label{th:oP-symm}  For a closed relation $A$ the following hold: 
 \begin{enumerate}[(i)]
 \item\label{it1:OpS} $A$ is symmetric if and only if it is d.m.o. and $\oP A$ is symmetric.
 \item\label{it2:OpS} $A$ is maximal symmetric if and only if it is densely d.m.o. with $A_A$ a maximal symmetric operator in ${(\mul A)^\perp}^{\oplus_2}$.
 \item\label{it3:OpS} $A$ is selfadjoint if and only if it is densely d.m.o. with $A_A$ a selfadjoint operator in ${(\mul A)^\perp}^{\oplus_2}$.
 \end{enumerate}
 \end{theorem}
\begin{proof}
\eqref{it1:OpS}: it readily follows by Lemma~\ref{lm:operator-part-L} and Proposition~\ref{pro:characterization-symmetric}.\eqref{it5:sym}.

\eqref{it2:OpS}: if $A$ is maximal symmetric then Remark~\ref{rmk:Asym-AA} guarantees that $A$ is densely d.m.o. Besides, Remark~\ref{rmk:dom-multperp-spectral}.\eqref{eq:Tt-ToP} implies that $A_A=\oP A\subset A$, i.e., $A_A$ is closed symmetric operator in ${(\mul A)^\perp}^{\oplus_2}$ and maximal from  Lemma~\ref{lem:max-Symmetric}, since  $\rho(A_A)=\rho(A)$ by virtue of  Remark~\ref{rmk:dom-multperp-spectral}.\eqref{eq:espectral-properties-inherited}. Conversely, if $A$ is densely d.m.o., with $A_A$ a maximal symmetric operator in ${(\mul A)^\perp}^{\oplus_2}$, then it follows from Remark~\ref{rmk:dom-multperp-spectral}.\eqref{eq:Tt-ToP} and item~\eqref{it1:OpS} that $A$ is closed symmetric and maximal, by Remark~\ref{rmk:dom-multperp-spectral}.\eqref{eq:espectral-properties-inherited} and  Lemma~\ref{lem:max-Symmetric}. 

\eqref{it3:OpS}: it is straightforward from item~\eqref{it2:OpS}, Remarks~\ref{rmk:dom-multperp-spectral}.\eqref{eq:espectral-properties-inherited}, \ref{rmk:Asym-AA} and  the second part of Proposition~\ref{pro:characterization-symmetric}.
\end{proof}

\section{Anti-dissipative and anti-accumulative linear relations}\label{s4}
We address in this section the notions of anti-dissipative and anti-accumulative linear relations, which, in the context of operators, are useful to characterize $C_0$-semigroups of contractions \cite[Sects.\,1.3 and 1.4]{MR710486}.

\begin{definition}\label{def:anti-accumulative}
A linear relation $S$ is said to be \emph{anti-dissipative} (\emph{anti-accumulative}) if
\begin{gather*}
\re\ip fg\geq0\quad (\re\ip fg\leq0)\,,\qquad \mbox{for all }\vE fg\in S\,.
\end{gather*}
Furthermore, $S$ is \emph{maximal anti-dissipative} (\emph{maximal anti-accumulative}) if it does not admit proper anti-dissipative (anti-accumulative) extensions. For instance, the maximal anti-dissipative and maximal anti-accumulative relations are closed. 
\end{definition}

A linear relation $S$ is anti-dissipative (maximal anti-dissipative) if and only if $-S$ is anti-accumulative (maximal anti-accumulative). This allows us to restrict our analysis to anti-dissipative relations, from which anti-accumulative results are immediately deduced.

\begin{lemma}\label{lem:nonA-Diss}
For a linear relation $S$ the following hold:
\begin{enumerate}[(i)]
\item\label{it1:nD} $S$ is anti-dissipative (anti-accumulative) if and only if $iS$ is dissipative (accumulative).
\item\label{it2:nD} $S$ is maximal anti-dissipative (maximal anti-accumulative) if and only if $iS$ is maximal dissipative (maximal accumulative).
\end{enumerate}
\end{lemma}
\begin{proof}
Note that $\vE fg\in S$ if and only if $\vE f{ig}\in iS$, and
\begin{gather*}
\re\ip fg=\re(-i\ip f{ig})=\im\ip f{ig}\,,
\end{gather*}
which yields \eqref{it1:nD}. Now, \eqref{it1:nD} is  straightforward since an extension $\hat S$ of $S$ means that $i\hat S$ is an extension of $iS$.
\end{proof}

In light of Lemma~\ref{lem:nonA-Diss}, we will use properties of dissipative (accumulative) relations to establish those of anti-dissipative (anti-accumulative) relations. For convenience, we denote the left and right open half-planes by 
\begin{gather*}
\C_{<0}\colonequals \{\lambda\in\C\,:\,\re \lambda <0\}\quad\mbox{and}\quad\C_{>0}\colonequals \{\lambda\in\C\,:\,\re \lambda >0\}\,,\quad\mbox{respectively}\,.
\end{gather*}

\begin{theorem}\label{th:A-Dissipative-resolvet} For a linear relation $S$ the following hold:
\begin{enumerate}[(i)]
\item  $S$ is anti-dissipative (maximal anti-dissipative) if and only if $\C_{<0}\subset \hat\rho(S)$ ($S$ is closed and $\C_{<0}\subset\rho(S)$) and $\no {\mathcal R_S(\lambda)}\leq -(\re \lambda)^{-1}$, for $\lambda\in\C_{<0}$.
\item\label{it2:aD-res}   $S$ is anti-accumulative (maximal anti-accumulative) if and only if $\C_{>0}\subset \hat\rho(S)$ ($S$ is closed and $\C_{>0}\subset\rho(S)$) and $\no {\mathcal R_S(\lambda)}\leq (\re \lambda)^{-1}$, for $\lambda\in\C_{>0}$.
\end{enumerate}
\end{theorem}
\begin{proof}
If $S$ is anti-dissipative then one has by Lemma~\ref{lem:nonA-Diss} that  $iS$ is dissipative. Thus, Theorem~\ref{th:Dissipative-resolvet} implies $\C_{<0}=-i\C_-\subset -i\hat\rho(iS)=\rho(S)$, while for $\lambda\in\C_{<0}$ it follows that $i\lambda\in\C_-$ and
\begin{gather*}
\no {\mathcal R_S(\lambda)}=\no {\mathcal R_{iS}(i\lambda)}\leq-\frac{1}{\im(i\lambda)}=-\frac{1}{\re\lambda}\,.
\end{gather*}
To the converse, one has that $\C_-=i\C_{<0}\subset  i\rho(S)=\rho(iS)$, whereas for $\zeta\in\C_-$ it fulfills $-i\zeta\in\C_{<0}$ and 
\begin{gather*}
\no{\mathcal R_{iS}(\zeta)}=\no{\mathcal R_{S}(-i\zeta)}\leq-\frac1{\re(-i\zeta)}=-\frac{1}{\im \zeta}\,.
\end{gather*}
Hence, Theorem~\ref{th:Dissipative-resolvet} and Lemma~\ref{lem:nonA-Diss} yield that $iS$ is dissipative and $S$ is anti-dissipative. The proof of the maximality condition follows along similar lines.
\end{proof}

\begin{remark}\label{rmk:aD-dmo}
 Every anti-dissipative (anti-accumulative) relation $S$ is d.m.o. Moreover, if $S$ is maximal anti-dissipative (anti-accumulative) then also is densely d.m.o. This follows from Remark~\ref{rm:dom-mul_Diss} and Lemma~\ref{lem:nonA-Diss}, since $\dom S=\dom iS$ and $\mul S=\mul iS$.
\end{remark}

It is a simple matter to check from Remark~\ref{rmk:aD-dmo} that  a closed anti-dissipative (anti-accumulative) relation whose domain is the whole space is a bounded maximal anti-dissipative (maximal anti-accumulative) operator. Besides, as a consequence of Theorem \ref{th:A-Dissipative-resolvet}, the spectrum of a maximal anti-dissipative (maximal anti-accumulative) relation $S$ holds
\begin{gather*}
\rho(S)\cap(\C_{<0}\cup i\R)=\hat\rho(S)\cap(\C_{<0}\cup i\R)\qquad (\rho(S)\cap(\C_{>0}\cup i\R)=\hat\rho(S)\cap(\C_{>0}\cup i\R))\,,
\end{gather*}
viz. $\sigma(S)\cup i\R=\hat\sigma(S)\cup i\R$, being $i\R\colonequals\{i\lambda\,:\, \lambda\in\R\}$.

Our attention now turns to the operator part of a closed anti-dissipative (anti-accumulative) relation.
\begin{theorem}\label{th:a-DA-oP}
For a closed linear relation $S$ the following hold:
\begin{enumerate}[(i)]
\item\label{it1:naOp} $S$ is anti-dissipative (anti-accumulative) if and only if it is d.m.o. such that $\oP S$ is anti-dissipative (anti-accumulative).
\item\label{it2:naOp} $S$ is maximal anti-dissipative (maximal anti-accumulative) if and only if it is densely d.m.o. with $S_S$ a maximal anti-dissipative (maximal anti-accumulative) operator in ${(\mul S)^{\perp}}^{\oplus_2}$.
\end{enumerate}
\end{theorem}
\begin{proof}
The proof of \eqref{it1:naOp} and \eqref{it2:naOp} readily follow from Lema~\ref{lm:operator-part-L} and Theorem~\ref{th:maximal-Ll}, respectively. This relies on Lemma~\ref{lem:nonA-Diss} and the identities $\mul S=\mul iS$, $\dom S=\dom iS$, and $\oP{(iS)}=i\oP S$.
\end{proof}

\subsection{Anti-symmetric and anti-selfadjoint linear relations} Building upon the concept of anti-dissipative relations (Definition~\ref{def:anti-accumulative}), we first introduce an anti-symmetric linear relation as follows: a relation $B$ is anti-symmetric if, for all $\vE fg\in B$, the condition $\re\ip{f}{g}=0$ holds. This can also be stated as $\im\ip{f}{ig}=0$, and it is equivalent to saying that $iB$ is a symmetric relation.

\begin{definition}
A relation $B$ is said to be \emph{anti-symmetric} if $B\subset -B^*$ and \emph{anti-selfadjoint} when $B=-B^*$. Moreover, an anti-symmetric relation B is call \emph{maximal}, if it does not admit
proper anti-symmetric extensions.
\end{definition}
	
\begin{remark}\label{rmk:equi-as-s} By definition, it is easy to see that $B$ and $\hat B$ are anti-symmetric and anti-selfadjoint if and only if $iB$ and $i\hat B$ are symmetric and selfadjoint, respectively. 
\end{remark}

The equivalence presented in Remark~\ref{rmk:equi-as-s} implies that the following result is a direct consequence of Proposition~\ref{pro:characterization-symmetric}. 	

\begin{proposition}\label{pro:characterization-A-symmetric}
For a linear relation $B$ the following are equivalent: 
\begin{enumerate}[(i)]
\item $B$ is anti-symmetric.
\item $\ip fg\in i\R$, for all $\vE fg\in B$.
\item $\ip fk=-\ip gh$, for all $\vE fg,\vE hk\in B$.
\item\label{it4:asym} $\C\backslash i\R\subset\hat\rho(B)$ and  
\begin{gather*}
\no{\mathcal R_B(\lambda)}\leq \abs{\re\lambda}^{-1}\,,\quad \mbox{for $\lambda\in\C\backslash i\R$}\,.
\end{gather*}

\item\label{it5:asym} $B$ is simultaneously anti-dissipative and anti-accumulative.
\end{enumerate}
Moreover, for a closed anti-symmetric relation $B$ the following properties are identical:
\begin{enumerate}[(a)]
\item $B$ is anti-selfadjoint.
\item\label{itb:a-sa} $\ker(B^*\pm I)=\{0\}$.
\item $\sigma(B)=\hat\sigma(B)$.
\end{enumerate}
\end{proposition}

Bearing in mind Theorem~\eqref{th:A-Dissipative-resolvet} one can define the deficiency index of a closed anti-dissipative relation $S$ and the deficiency index of a closed anti-accumulative relation $\hat S$ by 
\begin{gather*}
\eta_{<0}(L)\colonequals \eta_\lambda(S)\quad\mbox{and}\quad \eta_{>0}(\hat S)\colonequals \eta_{\gamma}(\hat S)\,,
\qquad \lambda,-\gamma\in\C_{<0}\,.
\end{gather*}
Besides, Proposition~\ref{pro:characterization-A-symmetric}.\eqref{it4:asym}-\eqref{it5:asym} implies that a closed anti-symmetric relation $B$ has deficiency indices  
\begin{gather*}
(\eta_{>0}(B),\eta_{<0}(B))\colonequals(\eta_{-\lambda}(B),\eta_{\lambda}(B))\,,\quad \lambda\in\C_{<0}\,,
\end{gather*}
which satisfy $\eta_{<0}(-B)=\eta_{>0}(B)$ and $\eta_{>0}(-B)=\eta_{<0}(B)$.

It is worth noting that the von Neumann formulae also hold for anti-symmetric  relations, which leads to the following result. We omit the proof, as it readily follows from Remark~\ref{rmk:equi-as-s} and the usual von Neumann formulae for symmetric relations \cite[Thms.\,4.6 and 4.7]{MR4082306}.

\begin{theorem}[von Neumann formulae for anti-symmetric relations]\label{thm:vN-formulae-AS} For a closed anti-symmetric relation $B$ the following hold:
\begin{enumerate}[(1\,$^\circ$)]
\item For $\lambda\in\C\backslash i\R$, 
\begin{gather*}
B^*= -B\dotplus \dS_{-\lambda}(B^*)\dotplus \dS_\lambda (B^*)\,,
\end{gather*}
whence the direct sum is orthogonal if $\lambda\in \{1,\,-1\}$.
\item A closed anti-dissipative (resp. anti-accumulative) relation $\hat B$ is an extension of $B$ if and only if for a fixed  $\lambda\in\C_{>0}$ (resp. $\lambda\in\C_{<0}$), 
\begin{gather}\label{von02-Ad}
  \hat B=B\dotplus (V- I)D\,,\end{gather} where
  $D\subset \dS_\lambda (B^*)$ is a closed bounded relation and
  $V:D\rightarrow \dS_{-\lambda} (B^*)$ is a closed
  contraction in $(\mathcal H^{\oplus_2})^{\oplus_2}$. For $\zeta=1$ (resp. $\zeta=-1$), the direct sum in \eqref{von02-Ad} is orthogonal.
  
  \item A closed anti-symmetric relation $\hat B$ is an extension of $B$ if and only if for a fixed $\lambda\in \C\backslash i\R$, \begin{gather}\label{von04-Ad}
  \hat B=B\dotplus (V- I)D\,,\end{gather} where
  $D\subset \dS_\lambda (B^*)$ is a closed bounded relation and
  $V:D\rightarrow \dS_{-\lambda} (B^*)$ is a closed isometry in $(\mathcal H^{\oplus_2})^{\oplus_2}$. For $\zeta=\pm1$, the direct sum in \eqref{von04-Ad} is orthogonal.  
\end{enumerate}
\end{theorem}

\begin{remark}\label{rm:A-dissipative-extension}
The extendability of a closed anti-symmetric relation $B$ is established by Theorem~\ref{thm:vN-formulae-AS} and Remark~\ref{rm:dissipative-extension}. Specifically, $B$ admits closed anti-dissipative extensions if and only if $\eta_{<0}(B)\neq0$, closed anti-accumulative extensions if and only if $\eta_{>0}(B)\neq0$, and closed anti-symmetric extensions if and only if both $\eta_{<0}(B),\eta_{>0}(B)\neq0$. Additionally, if $T$, $\hat T$ and $S$ denote closed anti-dissipative, anti-accumulative and anti-symmetric extensions of $B$, respectively, their deficiency indices satisfy:
\begin{equation}\label{eq:di-AS-AD}
\begin{aligned}
\eta_{<0}(B)&=\eta_{<0}(T)+\dim[T/B]\,, &\qquad \eta_{>0}(B)&=\eta_{>0}(\hat T)+\dim[\hat T/B]\,,\quad\mbox{while}\\ 
\eta_{<0}(B)&=\eta_{<0}(S)+\dim[S/B]\,, &\qquad \eta_{>0}(B)&=\eta_{>0}(S)+\dim[S/B]\,.
\end{aligned}
\end{equation}
\end{remark}

The following result is a consequence of Lemma~\ref{lem:max-Symmetric}, Remark~\ref{rmk:equi-as-s} and \eqref{eq:di-AS-AD}.

\begin{lemma}\label{lem:max-A-Symmetric}
For a closed anti-symmetric relation $B$ the following are equivalent:
\begin{enumerate}[(i)]
\item\label{it1:max-asym} $B$ is maximal.
\item\label{it2:max-asym} One of the indices $\eta_{>0}(B)$, $\eta_{<0}(B)$ is equal zero.
\item\label{it3:max-asym} $B$ maximal anti-dissipative or maximal anti-accumulative.
\item\label{it4:max-asym} $\C_{<0}$ or $\C_{>0}$ is contained in $\rho(B)$.
\end{enumerate}
\end{lemma}

Remark~\ref{rmk:aD-dmo} and Proposition~\ref{pro:characterization-A-symmetric}.\eqref{it5:asym} establish that anti-symmetric relations are d.m.o., with maximal ones being densely d.m.o. This implies that closed anti-symmetric relations exhibit the spectral properties \eqref{eq:espectral-properties-inherited} as well as the resolvent properties \eqref{eq:resolvent-OP} and \eqref{eq:rest-norm}.

\begin{theorem}\label{th:oP-asymm}  For a closed relation $B$ the following hold: 
 \begin{enumerate}[(i)]
 \item\label{it1:OpaS} $B$ is anti-symmetric if and only if it is d.m.o. and $\oP B$ is anti-symmetric.
 \item\label{it2:OpaS} $B$ is maximal anti-symmetric if and only if it is densely d.m.o. with $B_B$ a maximal anti-symmetric operator in ${(\mul B)^\perp}^{\oplus_2}$.
 \item\label{it3:OpaS} $B$ is anti-selfadjoint if and only if it is densely d.m.o. with $B_B$ a anti-selfadjoint operator in ${(\mul B)^\perp}^{\oplus_2}$.
 \end{enumerate}
 \end{theorem}
\begin{proof}
The proof is straightforward from Theorem~\ref{th:oP-symm} and Remark~\ref{rmk:equi-as-s}, utilizing the fact that $\mul B=\mul iB$, $\dom B=\dom iB$, and $\oP{(iB)}=i\oP B$.
\end{proof}

\section{Semigroups of linear relations}\label{s5}
In this section, we extend the notion of semigroups from linear operators to linear relations, based on maximal d.m.o. relations (see Definition~\ref{df:dmo-T}).We adapt several standard results concerning linear operators from \cite{MR710486}, which we will use freely throughout.

\begin{definition}\label{dfn:semigroup}
A one parameter family of closed maximal d.m.o. linear relations $\{T_t\}_{t\geq 0}$ is called \emph{semigroup} if there exists a closed subspace $\mathcal K\leq \mathcal H$ such that
\begin{enumerate}[(i)]
\item\label{it1:df-SG} $\mul T_t=\mathcal K$, for all $t\geq 0$\quad (the homogeneous d.m.o.).
\item $T_0=I\rE{\mathcal K^\perp}\oplus \lrb{\vE 0k\,:\, k\in \mathcal K}$\quad ($I\rE{\mathcal K^\perp}$ is identical operator on $\mathcal K^\perp$).
\item $T_{t+s}=T_tT_s$ for every $t,s\geq 0$\quad (the semigroup property).
\end{enumerate}
\end{definition}

From now on, we assume that a semigroup is a family of closed maximal d.m.o. linear relations, satisfying the conditions of Definition~\ref{dfn:semigroup}. 

\begin{remark}\label{rmk:Tt-BK}
A semigroup $\{T_t\}_{t\geq 0}$ satisfies $\dom T_t=(\mul T_t)^\perp=\mathcal K^\perp$. Thus, $\oP{T_t}$ is a closed operator with domain the whole space $\mathcal K^\perp$; that is, $\oP{T_t}\in \mathcal B(\mathcal K^\perp)$. In particular, when $\mathcal K={0}$, this reduces to the usual concept of a semigroup of bounded linear operators on $\mathcal H$.
\end{remark}

\begin{theorem}\label{th:oP-semigroup}
Let $\{T_t\}_{t\geq 0}$ be a family of maximal d.m.o. closed relations. Then $\{T_t\}_{t\geq 0}$ is semigroup if and only if  there exists a closed subspace $\mathcal K\leq \mathcal H$ such that
$\mul T_t=\mathcal K$, for all $t\geq 0$, and $\{\oP{T_t}\}_{t\geq0}$ is a semigroup of operators in $\mathcal B(\mathcal K^\perp)$.
\end{theorem}
\begin{proof}
If $\{T_t\}_{t\geq 0}$ is semigroup then there exists a closed subspace $\mathcal K\leq \mathcal H$ for which
$\mul T_t=\mathcal K$, while Remark~\ref{rmk:Tt-BK} states $\oP{T_t}\in \mathcal B(\mathcal K^\perp)$, for all $t\geq0$. Besides, $\oP{T_0}=I\rE{\mathcal K^\perp}$, while Theorem~\ref{prop:PC-tn}.\eqref{itb:PC} yields
\begin{gather*}
\oP{T_t}\oP{T_s}=\oP{(T_tT_s)}=\oP{T_{t+s}}\,,\qquad t,s\geq0\,,
\end{gather*}
whence it fulfills that $\{\oP{T_t}\}_{t\geq0}$ is a semigroup of operators in $\mathcal B(\mathcal K^\perp)$. Conversely, if  if  there exists a closed subspace $\mathcal K\leq \mathcal H$ for which
$\mul T_t=\mathcal K$, for all $t\geq 0$, and $\{\oP{T_t}\}_{t\geq0}\subset \mathcal B(\mathcal K^\perp)$ is a semigroup, then $T_0=\oP{T_0}\oplus {T_0}_\infty=I\rE{\mathcal K^\perp}\oplus {T_0}_\infty$, and again  Theorem~\ref{prop:PC-tn}.\eqref{itb:PC} produces
\begin{gather*}
T_tT_s=\oP{T_t}\oP{T_s}\oplus {T_0}_\infty=\oP{T_{t+s}}\oplus {T_{t+s}}_\infty=T_{t+s}\,,\qquad t,s\geq0\,,
\end{gather*}
i.e., $\{T_t\}_{t\geq 0}$ is semigroup.
\end{proof}

The following assertion guarantees that we can generate semigroups of linear relations by means of semigroups of linear operators.

\begin{corollary}\label{cor:Tt-sg-Ut}
Let $\mathcal K\leq\mathcal H$ be a closed subspace. If $\{U_t\}_{t\geq0}$ is a a semigroup of operators in $\mathcal B(\mathcal K^\perp)$, then the family $\{T_t\}_{t\geq0}$, where 
\begin{align}\label{eq:Tt-sg-Ut}
T_t=\lrb{\vE f{U_tf+k}\,:\,f\in \mathcal K^\perp\mbox{ and }k\in \mathcal K}\,,\qquad t\geq 0\,,
\end{align}
defines a semigroup of linear relations.
\end{corollary}
\begin{proof}
For $t\geq 0$, one readily checks that \eqref{eq:Tt-sg-Ut} is a closed linear relation with $\oP{T_t}=U_t$ and $\mul T_t=\mathcal K$. Besides, $\dom T_t=\dom U_t=\mathcal K^\perp=(\mul T_t)^\perp$. Hence, $\{T_t\}_{t\geq 0}$ is a family of maximal d.m.o. closed relations and Theorem~\ref{th:oP-semigroup} implies that it is a semigroup.
\end{proof}

By Theorem~\ref{th:oP-semigroup}, every semigroup of linear relations $\{T_t\}_{t\geq0}$ has associated an operator semigroup given by  $\{\oP{T_t}\}_{t\geq0}$. It is well-known that the infinitesimal generator of $\{\oP{T_t}\}_{t\geq0}$ is the closed linear operator 
\begin{gather}\label{eq:infi-gene-oP}
\oP A\colonequals \frac{d}{dt}\oP{T_t}\rE{t=0}=\lrb{\vE{u}{\frac{d}{dt}\oP{T_t}u\rE{t=0}}\,:\,\frac{d}{dt}\oP{T_t}u\rE{t=0}\mbox{ exits}}\leq {\mathcal K^\perp}^{\oplus_2}\,.
\end{gather}
Besides, due to  maximal d.m.o. property, $\dom \oP{T_t},\ran \oP{T_t}\subset (\mul T_t)^\perp=\mathcal K^\perp=(\mul T_0)^\perp$. Thereby, if we take the infinitesimal generator of the semigroup $\{T_t\}_{t\geq0}$ by 
\begin{align*}
A&\colonequals\frac{d}{dt}{T_t}\rE{t=0}=\lim_{t\to 0}\frac{T_t-T_0}{t}=\lim_{t\to 0}\frac{\oP{T_t}\oplus {T_0}_\infty-I\rE{\mathcal K^\perp}\oplus{T_0}_\infty}{t}
\\&=\lim_{t\to 0}\frac{(\oP{T_t}-I\rE{\mathcal K^\perp})\oplus{T_0}_\infty}{t}=
\lim_{t\to 0}\frac{(\oP{T_t}-I\rE{\mathcal K^\perp})}{t}\oplus{T_0}_\infty=\oP A\oplus {T_0}_\infty\,,
\end{align*}
then one can define the following.
\begin{definition}\label{dfn:generator-of-semigroup}
The \emph{infinitesimal generator} (or simply \emph{generator}) of a semigroup $\{T_t\}_{t\geq0}$ is the closed linear relation
\begin{gather}\label{eq:semigroup-generator}
A=\oP A\oplus {T_0}_\infty\,,
\end{gather}
being $\oP A$ the infinitesimal generator of the associated operator semigroup $\{\oP{T_t}\}_{t\geq0}$. Actually, the operator part of the infinitesimal generator \eqref{eq:semigroup-generator} is $\oP A$, while $A_\infty={T_0}_\infty$.
\end{definition}

\begin{remark}
The generator \eqref{eq:semigroup-generator} is d.m.o. and unique. Indeed, since 
\begin{gather*}
\dom A=\dom \oP A\subset (\mul T_0)^\perp=(\mul A)^\perp\,,
\end{gather*}
wherefrom $A$ is d.m.o. Besides, if $B$ is a generator of $\{T_t\}_{t\geq0}$, then 
\begin{gather*}
B=\oP B\oplus {T_0}_\infty=\frac{d}{dt}\oP{T_t}\rE{t=0}\oplus {T_0}_\infty=\oP A\oplus {T_0}_\infty=A\,.
\end{gather*}

\end{remark}

In what follows, we will address uniformly continuous and strongly continuous semigroups of linear relations.

\subsection{Uniformly Continuous Semigroups} For the study of uniformly continuous semigroups of linear relations, several standard facts concerning uniformly continuous semigroups of linear operators are adapted from \cite[Sect.\,1.1]{MR710486} and will be used freely. 

Recall by Theorem~\ref{th:oP-semigroup} that a semigroup $\{T_t\}_{t\geq 0}$ is homogeneous d.m.o. (see Definition~\ref{dfn:semigroup}.\eqref{it1:df-SG}) with $\{\oP{T_t}\}_{t\geq0}$ a semigroup in $\mathcal B((\mul T_0)^\perp)$.

\begin{definition}\label{dfn:unif-semiG}
A semigroup $\{T_t\}_{t\geq 0}$ is said to be \emph{uniformly continuous}  if the associated operator semigroup $\{\oP{T_t}\}_{t\geq0}$ is uniformly continuous in $\mathcal B((\mul T_0)^\perp)$, i.e., if 
\begin{gather*}
\lim_{t\to 0^+}\no{\oP{T_t}-I\rE{(\mul T_0)^\perp}}=0\,.
\end{gather*}

\end{definition}

The following result relies on the fact that a linear operator $S$ generates a unique uniformly continuous semigroup of operators if and only if $S$ is bounded \cite[Thm.\,1.2]{MR710486}.

\begin{theorem}\label{th:uniformly-CSG} A closed linear relation $A$ is the infinitesimal generator of a uniformly continuous semigroup $\{T_t\}_{t\geq0}$ if and only if $A$ is maximal d.m.o. In such a case the following statements hold true: 
\begin{enumerate}[(i)]
\item\label{it1:igATt} $\{T_t\}_{t\geq0}$ is the unique uniformly continuous semigroup generated by $A$.
\item\label{it2:igATt} $T_t=e^{t\oP A}\oplus {T_0}_\infty$. Hence, we can consider
\begin{gather}\label{eq:exp-mdmo_A}
e^{tA}\colonequals e^{t\oP A}\oplus {A}_\infty\,,\quad\mbox{whence}\quad T_t=e^{tA}\,,\quad t\geq0\,.
\end{gather}
\item\label{it3:igATt} For $\frac{d}{dt}T_t\colonequals \frac{d}{dt}{\oP{T_t}}\oplus {T_0}_\infty$, it follows that 
\begin{gather*}
\frac{d}{dt}T_t=AT_t=T_tA\,.
\end{gather*}
\end{enumerate}
\end{theorem}
\begin{proof}
If $A$ is the infinitesimal generator of a uniformly continuous semigroup $\{T_t\}_{t\geq0}$, then $\oP A$ is the infinitesimal generator of the uniformly continuous semigroup of operators $\{\oP{T_t}\}_{t\geq0}$, viz. $\oP A\in\mathcal B((\mul T_0)^\perp)$. So, one has by \eqref{eq:semigroup-generator} that $\dom A=\dom \oP A=(\mul T_0)^\perp=(\mul A)^\perp$, i.e., $A$ is maximal d.m.o. Conversely, the  maximal d.m.o. property of $A$ implies that  $\dom \oP A=(\mul T_0)^\perp$, whence $\oP A\in\mathcal B((\mul T_0)^\perp)$. In this fashion, $\oP A$ generates a uniformly continuous semigroup of operators $\{U_t\}_{t\geq0}$ and Corollary~\ref{cor:Tt-sg-Ut} readily yields that \eqref{eq:Tt-sg-Ut} is a uniformly continuous semigroup of linear relations, with generator $A$. Now, 

\eqref{it1:igATt}: if $\{S_t\}_{t\geq0}$ is a uniformly continuous semigroup generated by $A$, then the uniqueness of the generator $\oP A$ implies $\oP{S_t}=\oP{T_t}$, while the homogeneous d.m.o. condition yields $S_t=\oP{S_t}\oplus A_\infty=\oP{T_t}\oplus A_\infty=T_t$, for all $t\geq0$.

\eqref{it2:igATt}: due to the generator $\oP A$ of $\{\oP{T_t}\}_{t\geq0}$ satisfies (cf. \cite[Cor.\,1.4.(b)]{MR710486}) $\oP{T_t}=e^{\oP A}$, it follows that
$T_t=\oP {T_t}\oplus{T_0}_\infty=e^{t\oP A}\oplus {A}_\infty$, since ${T_0}_\infty=A_\infty$.

\eqref{it3:igATt}: following \cite[Cor.\,1.4.(d)]{MR710486}, it fulfills that $t\mapsto \oP{T_t}$ is differentiable in norm and $\frac{d}{dt}\oP{T_t}=\oP{A}\oP{T_t}=\oP{T_t}\oP{A}$. Thus, by virtue of $\mul A=\mul T_t$ and both $A $ and $T_t$ are maximal d.m.o., it follows from Theorem~\ref{prop:PC-tn}.\eqref{itb:PC} that $AT_t=\oP{A}\oP{T_t}\oplus {T_0}_\infty=\oP{T_t}\oP{A}\oplus {T_0}_\infty=T_tA$ and
\begin{gather*}
\frac{d}{dt}T_t=\frac{d}{dt}{\oP{T_t}}\oplus {T_0}_\infty=\oP{A}\oP{T_t}\oplus {T_0}_\infty=AT_t\,,
\end{gather*}
as required.
\end{proof}

\subsection{Strongly Continuous Semigroups} Following the approach of the preceding subsection, several standard facts about strongly continuous semigroups of linear operators, which are taken from \cite[Sect.\,1.2]{MR710486}, will be assumed and used freely.

\begin{definition}\label{dfn:st-semiG}
A semigroup $\{T_t\}_{t\geq 0}$ is said to be \emph{strongly continuous}  (briefly $C_0$-semigroup) if the associated operator semigroup $\{\oP{T_t}\}_{t\geq0}$ is strongly continuous in $\mathcal B((\mul T_0)^\perp)$, this means when 
\begin{gather*}
\lim_{t\to 0^+}\oP{T_t}f=f\,,\qquad f\in (\mul T_0)^\perp\,.
\end{gather*}
\end{definition}

\begin{remark}\label{rmk:properties-scSG}
If $\{T_t\}_{t\geq 0}$ is a $C_0$-semigroup with generator $A$, then the following properties hold (adapted from \cite[Cors.\,2.3, 2.5, and Thms.\,2.4, 2.6, 2.7]{MR710486}):
\begin{enumerate}[(i)]
\item For $f\in(\mul T_0)^\perp$, the function $t\mapsto \oP{T_t}f$ is continuous from $\mathbb R_{\geq0}$ into $(\mul T_0)^\perp$. 
\item $\oP A$ is a closed densely defined operator in $(\mul T_0)^\perp$.
\item The subspace $\cap_{n\geq1} \dom (\oP A)^n$ is dense in $(\mul T_0)^\perp$.
\item\label{it4:C0sgO} If $\{U_t\}_{t\geq 0}$ is a $C_0$-semigroup of operators in $(\mul T_0)^\perp$ with infinitesimal generator $\oP A$, then $U_t=\oP{T_t}$, for all $t\geq 0$.
\item For $f\in \dom \oP A$, it follows that $\oP{T_t}f\in\dom \oP A$ and
\begin{gather*}
\frac{d}{dt}\oP{T_t}f=\oP{A}\oP{T_t}f=\oP{T_t}\oP{A}f\,.
\end{gather*}
\end{enumerate}
\end{remark}

\begin{theorem}\label{th:prop-Co-Sg}
For a $C_0$-semigroup $\{T_t\}_{t\geq 0}$, with infinitesimal generator $A$, the following statements hold:
\begin{enumerate}[(i)]
\item\label{it1:scATt} $A$ is a closed densely d.m.o relation.
\item\label{it2:scATt} The subspace $\cap_{n\geq1} \dom A^n$ is dense in $(\mul T_0)^\perp$.
\item\label{it3:scATt} $\{T_t\}_{t\geq0}$ is the unique $C_0$-semigroup generated by $A$.
\item\label{it4:scATt} For $\frac{d}{dt}T_t\colonequals \frac{d}{dt}{\oP{T_t}}\oplus {T_0}_\infty$, it follows that 
\begin{gather}\label{eq:dtdTt-oP-scS}
\frac{d}{dt}T_t\rE{\dom A}=AT_t\rE{\dom A}=T_tA\,.
\end{gather}
\end{enumerate}
\end{theorem}
\begin{proof} Bearing in mind the properties of Remark~\ref{rmk:properties-scSG}:

\eqref{it1:scATt}: one has that $\cc{\dom A}=\cc{\dom \oP A}=(\mul T_0)^\perp=(\mul A)^\perp$, i.e., $A$ is densely d.m.o.

\eqref{it2:scATt}: the previous item and Theorem~\ref{prop:PC-tn}.\eqref{itb:PC} imply that $A^n=(\oP A)^n\oplus  {T_0}_\infty$, viz. 
\begin{gather*}
\oP{(A^n)}=(\oP A)^n\quad\mbox{and}\quad(A^n)_\infty= {T_0}_\infty\,. 
\end{gather*}
Therefore, $\cc{\cap_{n\geq1} \dom A^n}=\cc{\cap_{n\geq1} \dom \oP{(A^n)}}=\cc{\cap_{n\geq1} \dom (\oP A)^n}=(\mul T_0)^\perp$.

\eqref{it3:scATt}: if $\{S_t\}_{t\geq0}$ is a $C_0$-semigroup generated by $A$, then $\{\oP{S_t}\}_{t\geq0}$ is a $C_0$-semigroup generated by $\oP{A}$, which fulfills $\oP{S_t}=\oP{T_t}$. Thus, the homogeneous d.m.o. condition produces $S_t=\oP{S_t}\oplus A_\infty=\oP{T_t}\oplus A_\infty=T_t$, for all $t\geq0$.

\eqref{it2:scATt}: since $\mul A=\mul T_t$ and both $A $ and $T_t$ are maximal d.m.o., one has from Theorem~\ref{prop:PC-tn}.\eqref{itb:PC} that $AT_t\rE{\dom A}=\oP{A}\oP{T_t}\rE{\dom A}\oplus {T_0}_\infty=\oP{T_t}\oP{A}\oplus {T_0}_\infty=T_tA$. Besides, $\dom A=\dom \oP{A}$ yields 
\begin{align*}
\frac{d}{dt}T_t\rE{\dom A}&=\lrb{\vE f{\frac{d}{dt}\oP{T_t}f+g}\,:\,f\in\dom \oP{A}\,\mbox{ and }\,g\in\mul T_0}\\
&=\lrb{\vE f{\oP{A}\oP{T_t}f+g}\,:\,f\in\dom \oP{A}\,\mbox{ and }\,g\in\mul T_0}\\
&=\oP{A}\oP{T_t}\rE{\dom A}\oplus {T_0}_\infty=AT_t\rE{\dom A}\,,
\end{align*}
whence it follows \eqref{eq:dtdTt-oP-scS}.

\end{proof}

In what follows, we shall present characterizations of $C_0$-semigroups of linear relations. For instance, we state the following assertion, which holds for $C_0$-semigroups of operators (cf. \cite[Thm.\,5.3]{MR710486}).
\begin{theorem}\label{th:resl-C0}
A closed linear relation $A$ is the infinitesimal generator of a $C_0$-semigroup $\{T_t\}_{t\geq0}$ satisfying $\no{\oP{T_t}}\leq Me^{\omega t}$, with $M\geq1$ and $\omega\in\R$, if and only if $A$ is densely d.m.o., $(\omega,\infty)\subset\rho(A)$ and 
\begin{gather}\label{eq:resl-C0}
\no{\lrc{\mathcal R_A(\lambda)}^n}\leq \frac{M}{(\lambda-\omega)^n}\,,\qquad \lambda >\omega\,,\quad n\in\N\,.
\end{gather}
\end{theorem}
\begin{proof}
If $A$ is the generator of a $C_0$-semigroup $\{T_t\}_{t\geq0}$ satisfying $\no{\oP{T_t}}\leq Me^{\omega t}$, then the same holds for the infinitesimal generator $\oP A$ of the $C_0$-semigroup $\{\oP{T_t}\}_{t\geq0}$. In this fashion $\cc{\dom A}=\cc{\dom \oP A}=(\mul T_0)^\perp=(\mul A)^\perp$, whence Remark~\ref{rmk:dom-multperp-spectral} implies $(\omega,\infty)\subset A_A$ and
 \begin{gather*}
\no{\lrc{\mathcal R_{A_A}(\lambda)}^n}\leq \frac{M}{(\lambda-\omega)^n}\,.
\end{gather*}
Hence, $A$ is densely d.m.o., which \eqref{eq:espectral-properties-inherited} with Theorem~\ref{th:resolvet-dom-mul} imply $(\omega,\infty)\subset\rho(A)$ and \eqref{eq:resl-C0}. Conversely, the densely d.m.o. property of $A$, Remark~\ref{rmk:dom-multperp-spectral} and Theorem~\ref{th:resolvet-dom-mul} yield that $\oP A=A_A$ is a closed densely defined operator in $(\mul A)^\perp$, which satisfies  $(\omega,\infty)\subset\rho(A)=\rho(A_A)$ and $\no{\lrc{\mathcal R_{A_A}(\lambda)}^n}=\no{\lrc{\mathcal R_A(\lambda)}^n}\leq M(\lambda-\omega)^{-n}$. Therefore, $\oP{A}$ is the generator of a $C_0$-semigroup $\{\oP{T_t}\}_{t\geq0}$ of operators in $\mathcal B((\mul A)^\perp)$, satisfying $\no{\oP{T_t}}\leq Me^{\omega t}$, wherefrom Corollary~\ref{cor:Tt-sg-Ut} implies the assertion.
\end{proof}

A well-known and particularly useful tool for characterizing $C_0$-semigroups of operators is the \emph{Yosida approximation} of a densely defined closed linear operator $S$ in a Hilbert space $\mathcal K$, defined as
\begin{gather}\label{eq:Yosida-A}
S(\lambda)\colonequals \lambda S\mathcal R_S(\lambda)\in\mathcal B(\mathcal K)\,,\qquad \lambda\in\rho(S)\,.
\end{gather}
\begin{remark}\label{rmk:YAO}
In light of Theorem~\ref{th:resl-C0} applied to operators, if an operator $S$ generates a $C_0$-semigroup $\{U_t\}_{t\geq0}$, then $U_t=\lim_{\lambda \to \infty} e^{-tS(\lambda)}$ in the strong sense of operators (cf. \cite[Thm.\,5.5]{MR710486}). Since $S(\lambda)$ converges strongly to $-S$ as $\lambda\to \infty$ (cf. \cite[Lem.\,3.3]{MR710486}), it is customary to shorten the notation and write $e^{tS}\colonequals\lim_{\lambda \to \infty}e^{-tS(\lambda)}$ whenever $S$ is the generator of a $C_0$-semigroup. Thus,
\begin{gather*}
U_t=e^{tS},,\qquad t\geq0\,.
\end{gather*}
\end{remark}

\begin{theorem}\label{th:exp-genA}
If $A$ is the infinitesimal generator of a $C_0$-semigroup $\{T_t\}_{t\geq0}$ then 
\begin{gather*}
T_t=e^{t\oP A}\oplus {A}_\infty\,,\qquad t\geq0\,.
\end{gather*}
Therefore, we can denote in this context
\begin{gather}\label{eq:exp-genA}
e^{tA}\colonequals e^{t\oP A}\oplus {A}_\infty\,,\quad\mbox{whence}\quad T_t=e^{tA}\,,\quad t\geq0\,.
\end{gather}

\end{theorem}
\begin{proof}
If $A$ is the generator of a $C_0$-semigroup $\{T_t\}_{t\geq0}$ then so is $\oP A$ of the $C_0$-semigroup $\{\oP{T_t}\}_{t\geq0}$, which by Remark~\ref{rmk:YAO} one has that  
$\oP{T_t}= e^{t\oP A}$. Hence, it follows for $t\geq 0$ that 
\begin{gather*}
T_t={\oP T}_t\oplus {T_0}_\infty=e^{tA}\,,
\end{gather*}
since ${T_0}_\infty=A_\infty$.
\end{proof}

When $\oP{A}\in\mathcal B((\mul A)^\perp)$ in Theorem~\ref{th:exp-genA}, its Yosida approximation $\oP{A}(\lambda)$ converges uniformly to $-\oP A$ as $\lambda\to \infty$ (see Remark~\ref{rmk:YAO}). As a result, expression~\eqref{eq:exp-genA} coincides with~\eqref{eq:exp-mdmo_A}.

\subsection{The Hille-Yosida and Lumer Phillips theorems on linear relations}\label{ss5.3} 
We now address the analogous characterization for $C_0$-semigroups of contractions, which indeed relies on anti-accumulative operators \cite[Sects.\,1.3 and 1.4]{MR710486}. For the reader's convenience, we adapt the following result from \cite[Thms.\,3.1 and 4.3]{MR710486}.
\begin{theorem}[the Hille-Yosida and Lumer-Phillips theorems for operators]\label{th:HY-LP-o}
For a closed densely defined linear operator $S$ in a Hilbert space $\mathcal K$, the following statements are equivalent:
\begin{enumerate}[(i)]
\item $S$ is the infinitesimal generator of a $C_0$-semigroup of contractions
\item\label{it2:C0-sO} $(0,\infty)\subset\rho(S)$ and $\no{\mathcal R_S(\lambda)}\leq\lambda^{-1}$, for $\lambda>0$.
\item\label{it3:C0-sO} $\ran (S-\lambda I)=\mathcal H$ for all $\lambda >0$ and $S$ is anti-accumulative.
\end{enumerate}
\end{theorem}

By virtue of Theorem~\ref{th:A-Dissipative-resolvet}.\eqref{it2:aD-res}, conditions \eqref{it2:C0-sO} and \eqref{it3:C0-sO} of Theorem~\ref{th:HY-LP-o} are equivalent to saying that $S$ is maximal anti-accumulative. Therefore, Theorem~\ref{th:HY-LP-o} can be stated as follows.

\begin{corollary}\label{cor:HY-LP-o}
A closed densely defined linear operator $S$ is the infinitesimal generator of a $C_0$-semigroup of contractions if and only if it is maximal anti-accumulative.
\end{corollary}

We now present the Hille-Yosida and Lumer-Phillips theorems for linear relations.

\begin{theorem}\label{thm:HY-LP-lr}
A closed linear relation $S$ is the infinitesimal generator of a $C_0$-semigroup $\{T_t\}_{t\geq0}$ such that its operator part $\oP{T_t}$ is a contraction for all $t\geq0$, if and only if $S$ is maximal anti-accumulative.
\end{theorem}
\begin{proof}
If $S$ is the generator of a $C_0$-semigroup $\{T_t\}_{t\geq0}$ such that $\oP{T_t}$ is a contraction for all $t\geq0$, then $\oP{S}$ generates the $C_0$-semigroup of contractions $\{\oP{T_t}\}_{t\geq0}$, and $S$ is densely d.m.o., by Theorem~\ref{th:resl-C0}. Therefore, Corollary~\ref{cor:HY-LP-o} implies that $\oP{S}=S_S$ is maximal anti-accumulative in ${(\mul S)^\perp}^{\oplus_2}$, while Theorem~\ref{th:a-DA-oP}.\eqref{it2:naOp} guarantees that $S$ is maximal anti-accumulative.

Conversely, if $S$ is maximal anti-accumulative then Theorem~\ref{th:a-DA-oP}.\eqref{it2:naOp} yields that $S$ is densely d.m.o., being $\oP{S}$ a maximal anti-accumulative operator in ${(\mul S)^\perp}^{\oplus_2}$. Hence, Theorem~\ref{th:resl-C0} and Corollary~\ref{cor:HY-LP-o} imply that $S$ generates a $C_0$-semigroup $\{T_t\}_{t\geq0}$ such that $\oP{T_t}$ is a contraction for all $t\geq0$.
\end{proof}

The following result uses the well-known fact that for a self-adjoint operator $A$, $iA$ is the infinitesimal generator of a $C_0$-semigroup of unitary operators (cf. \cite[Sect.\,7.6]{MR566954}).

\begin{corollary}\label{cor:HY-LP-aA}
A closed linear relation $S$ is the generator of a $C_0$-semigroup $\{T_t\}_{t\geq0}$ such that $\oP{T_t}$ is unitary for all $t\geq0$, if and only if $S$ is anti-selfadjoint.
\end{corollary}
\begin{proof}
If $S$ generates a $C_0$-semigroup $\{T_t\}_{t\geq0}$ for which $\oP{T_t}$ is unitary for all $t\geq0$, then one has by \eqref{eq:infi-gene-oP} that  $\oP{S}, \oP{-S}$ generate the $C_0$-semigroups of unitary operators $\{\oP{T_t}\}_{t\geq0}, \{\oP{T_t}^{-1}\}_{t\geq0}$ in ${(\mul S)^{\perp}}^{\oplus_2}$, which certainly $\oP{T_t},\oP{T_t}^{-1}\in\mathcal B((\mul S)^{\perp})$, and $S,-S$ generate the $C_0$-semigroup $\{T_t\}_{t\geq0}, \{\oP{T_t}^{-1}\oplus {T_0}_\infty\}_{t\geq0}$, respectively. Since unitary operators are contraction, it follows by Theorem~\ref{thm:HY-LP-lr} that $S,-S$ are maximal anti-accumulative. Thus, Theorem~\ref{th:A-Dissipative-resolvet} implies $\pm 1\in\rho(S)$, i.e., $\mathcal R_S(\pm 1)\in\mathcal B(\mathcal H)$ and the last equality of \eqref{eq:p-Adjoint} yields
\begin{gather*}
\ker(S^*\pm I)=\lrc{\ran (S\pm I)^{-1}}^\perp=\lrc{\dom \mathcal R_S(\pm 1)}^\perp=\mathcal H^\perp=\{0\}\,.
\end{gather*}
Therefore, the second part of Proposition~\ref{pro:characterization-A-symmetric} implies that $S$ is anti-selfadjoint.

Conversely, if $S$ is anti-selfadjoint, then Theorem~\ref{th:A-Dissipative-resolvet} and Proposition~\ref{pro:characterization-A-symmetric} imply that $S$ is maximal anti-accumulative and, consequently, generates a $C_0$-semigroup $\{T_t\}_{t\geq0}$ according to Theorem~\ref{thm:HY-LP-lr}. Moreover, Theorem~\ref{th:oP-asymm}.\eqref{it3:OpaS} and Remark~\ref{rmk:equi-as-s} ensure that both $\oP S$ and $-\oP S$ are anti-selfadjoint, meaning that $-i\oP S$ is selfadjoint in ${(\mul S)^{\perp}}^{\oplus_2}$. Therefore, $\oP{S}=i(-i\oP S)$ is the generator of a $C_0$-semigroup of unitary operators, which is precisely $\{\oP {T_t}\}_{t\geq0}$ by Remark~\ref{rmk:properties-scSG}.\eqref{it4:C0sgO}.
\end{proof}

We conclude with the following result, which provides an analog of Stone's theorem for semigroups of linear relations.

\begin{corollary}\label{cor:Stone-generalization}
A family of closed maximal d.m.o. linear relations $\{T_t\}_{t\geq 0}$ is a  $C_0$-semigroup such that $\oP{T_t}$ is a contraction (unitary) for all $t\geq0$, if and only if there exists a uniquely determined maximal dissipative (self-adjoint) relation $L$ for which
\begin{gather}\label{eq:C0-representation}
T_t=e^{it\oP L}\oplus L_\infty=e^{itL}\,.
\end{gather}
Besides, $e^{it\oP L}=\lim_{\lambda \to \infty} e^{-it\oP L(-i\lambda)}$, as in the Remark \ref{rmk:YAO}.
\end{corollary}
\begin{proof}
If $\{T_t\}_{t\geq 0}$ is a  $C_0$-semigroup such that $\oP{T_t}$ is a contraction (unitary) for all $t\geq0$, then by Theorem~\ref{thm:HY-LP-lr} (Corollary~\ref{cor:HY-LP-aA}) its generator $S$ is maximal anti-accumulative (anti-selfadjoint). Besides, Lemma~\ref{lem:nonA-Diss} (Remark~\ref{rmk:equi-as-s}) yields that $L=-iS$ is maximal dissipative (selfadjoint). It is straightforward to check the identities 
\begin{gather}\label{eq:eqiL-S}
\oP L=-i\oP S\,;\quad L_\infty=S_\infty={T_0}_\infty\,,\quad\mbox{as well as}\quad (iL)_\infty=L_\infty\,; \quad i\oP L=\oP{(iL)}\,.
\end{gather}
Therefore, Theorem~\ref{th:exp-genA} yields
\begin{align*}
T_t&=e^{t\oP S}\oplus S_\infty=e^{it(-i\oP S)}\oplus L_\infty=e^{it\oP L}\oplus L_\infty\\
&=e^{t\oP{(iL)}}\oplus (iL)_\infty=e^{itL}\,.
\end{align*}
Moreover, the Yosida approximation \eqref{eq:Yosida-A} follows 
\begin{gather}\label{eq:Yosida-Av2}
\oP S(\lambda)=\lambda \oP S(\oP S-\lambda I)^{-1}=i(-i\lambda)(-i\oP S)(-i\oP S-(-i\lambda) I)^{-1}=i\oP L(-i\lambda)\,.
\end{gather}
Hence, Remark~\ref{rmk:YAO} and \eqref{eq:Yosida-Av2} yield
\begin{gather*}
e^{it\oP L}=e^{t\oP S}=\lim_{\lambda \to \infty}e^{-t\oP S(\lambda)}=\lim_{\lambda \to \infty}e^{-it\oP L(-i\lambda)}\,.
\end{gather*}
The uniqueness of $L$ follows from the uniqueness of the generator $S$.

Conversely, if $L$ is a maximal dissipative (selfadjoint) which satisfies \eqref{eq:C0-representation}, then due to Lemma~\ref{lem:nonA-Diss} (Remark~\ref{rmk:equi-as-s}), $S=iL$ is maximal anti-accumulative (anti-selfadjoint). Thereby, Theorem~\ref{thm:HY-LP-lr} (Corollary~\ref{cor:HY-LP-aA}) implies that $S$ is the generator of a $C_0$-semigroup $\{\hat T_t\}_{t\geq 0}$ such that $\oP{\hat T_t}$ is a contraction (unitary) for all $t\geq0$. Therefore, the equality $-iS=L$, the identities \eqref{eq:eqiL-S} and Theorem~\ref{th:exp-genA} yield $\hat T_t= T_t$, for all $t\geq0$.
\end{proof}

\section{$C_0$-semigroup generators as one-dimensional perturbations}
\subsection{One-dimensional perturbations of anti-selfadjoint relations}\label{subs-6.1}
Let $S\leq \mathcal H^{\oplus_2}$ be an anti-selfadjoint relation, $\varphi\in\dom S$ a non-zero normalized element and $Z\colonequals \Span\lrb{\vE 0\varphi}$.  One has by Theorem~\ref{th:oP-asymm}.\eqref{it3:OpaS}  that $S$ is densely d.m.o. and $S,Z$ are l.i. Thus, the relation
\begin{gather*}
B=S\cap Z^*\leq S\,,
\end{gather*}
is closed and anti-symmetric. Besides, $Z$ is unidimensional, anti-symmetric, and satisfies $Z=-Z$. Thus by \eqref{eq:p-Adjoint},
\begin{equation}\label{eq:adjoint-B-As}
\begin{aligned}
B^*&=\left(S\cap Z^*\right)^*=-\left(\left(S\cap Z^*\right)^\perp\right)^{-1}=-\left(\cc{S^\perp\dotplus (Z^*)^\perp}\right)^{-1}\\&=\cc{(-S^{-1})^\perp\dotplus(- (Z^*)^{-1})^\perp}=\cc{S^*\dotplus Z}=-S\dotplus Z\,.
\end{aligned}
\end{equation}
Thence, \eqref{rm:A-dissipative-extension} and Proposition~\ref{pro:characterization-A-symmetric}.\eqref{itb:a-sa} imply that the indices of $B$ are equal. Since $S$ is anti-selfadjoint, $(S-I)^{-1}\in\mathcal B(\mathcal H)$, i.e., there exits $\vE hk\in S$ such that $\varphi=k-h$ and by \eqref{eq:adjoint-B-As},
\begin{gather*}
\vE h{-h}=\vE h{-k+\varphi}\in  B^*\,,
\end{gather*}
where $h\neq 0$, since otherwise $\vE 0\varphi\in S$, a contradiction with $S\cap Z=\lrb{\vE 00}$. Besides, for $\vE u{-u}\in B^*$ one has by \eqref{eq:adjoint-B-As} that there exists $\vE ug\in S$ and $\alpha\in\C$ such that $-u=-g+\alpha\varphi$. Thus, 
\begin{gather*}
\vE{u-\alpha h}{u-\alpha h}=\vE ug-\alpha \vE hk\in S\,,
\end{gather*}
whence Proposition~\ref{pro:characterization-A-symmetric}.\eqref{itb:a-sa} yields $u=\alpha h$ and $\ker \dS_{-1}(B^*)=1$. Hence, $\eta_{>0}(B)=\eta_{<0}(B)=1$. 

\begin{theorem}\label{th:od-PaS}
All maximal anti-dissipative and maximal anti-accumulative extensions of $B$ are in one-to-one correspondence with $\C\cup\{\infty\}$. These extensions, denoted $S^{(\tau)}$, are one-dimensional perturbations of $S$ given by
\begin{gather}\label{eq:non-inf-Lt}
S^{(\tau)}=\lrb{\vE{f}{g+\tau\ip{\varphi}f \varphi}\,:\,\vE fg\in S}\,,\quad \tau\neq\infty\,,
\end{gather}
while for $\tau=\infty$:
\begin{gather}\label{eq:inf-Lt}
S^{(\infty)}=B\dotplus Z\,.
\end{gather}
Furthermore, $S^{(\tau)}$ is anti-dissipative if $\tau\in\C_{>0}$, anti-accumulative if $\tau\in\C_{<0}$, and anti-selfadjoint if $\tau\in i\R\cup\{\infty\}$.
\end{theorem}
\begin{proof}
Since $B\subset S$ and for $\vE hk\in B$, one has $\vE hk\in Z^*$ and $\ip{h}{\varphi}=\ip{k}{0}=0$, i.e., $\dom B\perp \varphi$ and $B\subset S^{(\tau)}$. Besides, it a simple matter to verify that $S^{(\tau)}$ is anti-dissipative if $\tau\in\C_{>0}$, anti-accumulative if $\tau\in\C_{<0}$, and anti-symmetric if $\tau\in i\R\cup\{\infty\}$. The maximality of each $S^{(\tau)}$ is a direct consequence of the indices presented in \eqref{eq:di-AS-AD} and this also reveals that $S^{(\tau)}$ is anti-selfadjoint when $\tau\in i\R\cup\{\infty\}$. 

It remains to prove that any anti-dissipative, respectively anti-accumulative, extensions of $B$ has the indicated form. If $L$ is a maximal anti-dissipative extensions of $B$, then the formulae of Theorem~\ref{thm:vN-formulae-AS} that $L\subset -B^*$ and \eqref{eq:adjoint-B-As} implies that for $\vE fh\in L$, there exists $\vE fg\in S$ and $\alpha\in\C$ such that $h=g+\alpha \varphi$. By virtue of $B\subset L$, assume that $\vE fg\in S$ does not belong to $B$. So, for $\ip \varphi f\neq0$, one may take $\tau=\alpha/\ip \varphi f$ and it follows that
\begin{gather*}
\vE fh=\vE f{g+\tau\ip \varphi f\varphi}\in S^{(\tau)}\,,
\end{gather*}
whence $0\leq\re\ip fh=\re (\tau \abs{\ip f\varphi}^2)$, i.e, $\re\tau \geq0$. Therefore, $L\subset S^{(\tau)}$ and equals, due to maximality.  When $\ip \varphi f=0$, one has that $\vE fg\in B$, viz.  $L= S^{(\infty)}$. The case when $L$ is a maximal anti-accumulative extensions of $B$ is straightforward, bearing in mind that $-L$ is a maximal anti-dissipative extensions of the closed anti-symmetric relation $-B$.
\end{proof}

Theorem~\ref{th:od-PaS} presents a variation of the result concerning one-dimensional perturbations of selfadjoint relations, analogous to that found in \cite[Thm.\,2.4]{MR1430397} (cf. \cite[Thm.\,4.5]{RC2025PADI}).

The following result is straightforward from Theorems~\ref{th:exp-genA},~\ref{thm:HY-LP-lr},~\ref{th:od-PaS} and Corollary~\ref{cor:HY-LP-aA}.

\begin{corollary}
For a fixed $\tau\in\C_{<0}\cup i\R\cup\{\infty\}$, the extension $S^{(\tau)}$ given in \eqref{eq:non-inf-Lt}-\eqref{eq:inf-Lt} generates the $C_0$-semigroup $\{T_t\}_{t\geq 0}$, where
\begin{gather*}
T_t=e^{tS^{(\tau)}}=e^{t\oP S^{(\tau)}}\oplus S^{(\tau)}_\infty\,,\qquad t\geq0\,.
\end{gather*}
In this expression, $e^{t\oP S^{(\tau)}}\in\mathcal B((\mul S^{(\tau)})^\perp)$ is a contraction, becoming unitary if $\tau\in i\R\cup\{\infty\}$.
\end{corollary}

\subsection{One-dimensional schr\"odinger generators}
For a fixed $0<a<\infty$, let $L_2(0,a)$ and $L_1(0,a)$ denote the spaces of square-integrable and absolutely-integrable functions, respectively, with respect to the Lebesgue measure on $(0,a)$. We also consider the space of smooth functions with compact support $C_{0}^{\infty}(0,a)$ and the absolutely continuous functions $AC[\alpha,\beta]$ for $0<\alpha<\beta<a$.

Consider the differential Schr\"odinger operator on $L_2(0,a)$ given by
\begin{gather*}
H\colonequals-\frac{d^{2}}{dx^{2}}+V(x)\,,
\end{gather*}
where $V(x)$ is a real-valued function in $L_1(0,a)$, and define the maximal domain $M$ as
\begin{align*}
M\colonequals\{f\in L_{2}(0,a)\,:\,& f, f'\in AC[\alpha,\beta]\,\mbox{ for every }[\alpha,\beta]\subset (0,a)\,,\\ &\mbox{ and }\, Hf\in L_{2}(0,a)\}\,.
\end{align*} 

Let $A$ be the operator $H$ restricted to $\dom A$, where
\begin{align*}
\dom A\colonequals\{ f\in M\,:\, f'(0)=f(a)=f'(a)=0\}\,.
\end{align*}
Then $A$ is a densely defined closed symmetric operator in $L_2(0,a)$, possessing deficiency indices $\eta_{\pm}(A)=1$. Its adjoint $A^*$ is given by $A^*=H\rE{\dom{A^*}}$, with
\begin{gather*}
\dom{A^*}=\{f\in M\,:\, f'(0)=0\}\,.
\end{gather*}
Furthermore, all maximal dissipative extensions $L_{\tau}$ of $A$ are densely defined in $L_2(0,a)$. These extensions are parametrized by $\kappa\in\C_{+}\cup\R\cup\{\infty\}$ and are given by $L_{\kappa}=H\rE{\dom{L_\kappa}}$, where
\begin{align}\label{podeSo}
\dom{L_\kappa}=\lrb{f\in M\,:\,f'(0)=0\,, f(a)=\kappa f'(a)}\,,\quad \text{for }\tau\neq\infty\,,
\end{align}
and for $\kappa=\infty$:
\begin{align}\label{eq:podeSo-inf}
\dom{L_\infty}=\lrb{f\in M\,:\,f'(0)=f'(a)=0}\,.
\end{align}
Among these maximal dissipative extensions, $L_{\kappa}$ is selfadjoint precisely when $\kappa\in\R\cup\{\infty\}$. This construction is adapted from \cite[Subsect.\, 4.2]{RC2025PADI} (see also \cite[Chap.\,15]{MR2953553}).

Now, for a fixed $\kappa\in\R\cup\{\infty\}$, let us consider a non-zero normalized element $\varphi\in\dom L_\kappa$. We define the relations $Z\colonequals\Span\lrb{\vE 0\varphi}$ and  $A_\kappa=L_\kappa\cap Z^*\leq L_\kappa$. This $A_\kappa$ can be realized as an operator $H\rE{\dom A_{\kappa}}$, where its domain is given by
\begin{gather*}
\dom{A_\kappa}=\lrb{f\in \dom L_\kappa\,:\, f\perp\varphi}\,.
\end{gather*}
This construction shows that $A_\kappa$ is a closed, non-densely defined symmetric operator. Consequently, since $-iL_\kappa$ is anti-selfadjoint (cf. Remark~\ref{rmk:equi-as-s}), we can adapt the results from Subsection~\ref{subs-6.1} to the following theorem.

\begin{theorem}
The deficiency indices of $A_\kappa$ are $\eta_\pm(A_\kappa)=1$. Besides, all maximal dissipative and maximal accumulative extensions of $A_\kappa$ are in one-to-one correspondence with $\C\cup\{\infty\}$. These extensions, denoted $L_\kappa^{(\tau)}$, are all operators and one-dimensional perturbations of $L_\kappa$ given by
\begin{gather}\label{eq:non-inf-Lk}
L_\kappa^{(\tau)}=\lrb{\vE{f}{g+\tau\ip{\varphi}f \varphi}\,:\,\vE fg\in L_\kappa}\,,\quad \tau\neq\infty\,,
\end{gather}
while $L_\kappa^{(\infty)}$ is the unique linear relation with non-trivial multivalued part, given by 
\begin{gather}\label{eq:inf-Lk}
L_\kappa^{(\infty)}=A_\kappa\dotplus Z\,.
\end{gather}
Furthermore, $L_\kappa^{(\tau)}$ is dissipative if $\tau\in\C_{+}$, accumulative if $\tau\in\C_{-}$, and selfadjoint if $\tau\in \R\cup\{\infty\}$.
\end{theorem}

Combining the preceding theorem with Theorem~\ref{th:exp-genA} and Corollary~\ref{cor:Stone-generalization}, we arrive at the following corollary.

\begin{corollary}
For $\kappa\in\R\cup\{\infty\}$ and $\tau\in\C_{+}\cup \R\cup\{\infty\}$, the extension $L_\kappa^{(\tau)}$ given in \eqref{eq:non-inf-Lk}-\eqref{eq:inf-Lk} generates the $C_0$-semigroup 
\begin{gather*}
\lrb{e^{itL_\kappa^{(\tau)}}}_{t\geq0}\,.
\end{gather*}
Besides, $e^{itL_\kappa^{(\tau)}}\in\mathcal B(L_2(0,a))$ is a contraction for $\tau\in\C_{+}$, and unitary if $\tau\in \R$. Moreover, the operator part of $e^{it{L_\kappa^{(\infty)}}}$, acting on $L_2(0,a)\ominus\{\varphi\}$, is a unitary operator.
\end{corollary}

\subsection*{Acknowledgment}
This work was partially supported by CONACYT-Mexico Grant CBF2023-2024-1842. 

\def\cprime{$'$} \def\lfhook#1{\setbox0=\hbox{#1}{\ooalign{\hidewidth
  \lower1.5ex\hbox{'}\hidewidth\crcr\unhbox0}}} \def\cprime{$'$}
  \def\cprime{$'$} \def\cprime{$'$} \def\cprime{$'$} \def\cprime{$'$}
  \def\cprime{$'$} \def\cprime{$'$}
\providecommand{\bysame}{\leavevmode\hbox to3em{\hrulefill}\thinspace}
\providecommand{\MR}{\relax\ifhmode\unskip\space\fi MR }
\providecommand{\MRhref}[2]{%
  \href{http://www.ams.org/mathscinet-getitem?mr=#1}{#2}
}
\providecommand{\href}[2]{#2}

\end{document}